\documentclass[12pt]{article}
\usepackage[left=3cm,top=3cm,right=3cm,bottom=3cm]{geometry}
\usepackage{geometry}
\usepackage{amsmath,amsfonts,amssymb}
\usepackage{natbib}
\usepackage{changepage}
\usepackage{graphicx}

\usepackage{url}
\usepackage[colorlinks=true, a4paper=true, pdfstartview=FitV,linkcolor=blue, citecolor=blue,urlcolor=blue]{hyperref}

\def\T{{ \mathrm{\scriptscriptstyle T} }}
\newcommand{\B}{\textbf}

\newcommand{\tD}{\tilde}
\newcommand{\M}{\mbox}
\DeclareMathOperator{\Diag}{\mathbf{Diag}}

\DeclareMathOperator{\tr}{\mathbf{tr}}

\numberwithin{equation}{section}

\newtheorem{theorem}{Theorem}[section]
\newtheorem{lemma}[theorem]{Lemma}
\newtheorem{proposition}[theorem]{Proposition}
\newtheorem{corollary}[theorem]{Corollary}

\newenvironment{proof}[1][Proof]{\begin{trivlist}
\item[\hskip \labelsep {\bfseries #1}]}{\end{trivlist}}

\newenvironment{assumption}[1][Assumption]{\begin{trivlist}
\item[\hskip \labelsep {\bfseries #1}]}{\end{trivlist}}

\newcommand{\qed}{\nobreak \ifvmode \relax \else
      \ifdim\lastskip<1.5em \hskip-\lastskip
      \hskip1.5em plus0em minus0.5em \fi \nobreak
      \vrule height0.75em width0.5em depth0.25em\fi}

\usepackage{authblk}

\title{Reduced rank regression via adaptive nuclear norm penalization}
\author[1]{Kun Chen\thanks{Corresponding author; \href{mailto:kunchen@ksu.edu}{kunchen@ksu.edu}}}
\author[2]{Hongbo Dong}
\author[3]{Kung-Sik Chan}
\affil[1]{Department of Statistics, Kansas State University}
\affil[2]{Wisconsin Institutes for Discovery, University of Wisconsin Madison}
\affil[3]{Department of Statistics and Actuarial Science, University of Iowa}

\date{\today}

\begin{document}

\maketitle

\begin{abstract}
Adaptive nuclear-norm penalization is proposed for low-rank matrix approximation, by which we develop a new reduced-rank estimation method for the general high-dimensional multivariate regression problems. The adaptive nuclear norm of a matrix is defined as the weighted sum of the singular values of the matrix. For example, the pre-specified weights may be some negative power of the singular values of the data matrix (or its projection in regression setting). The adaptive nuclear norm is generally non-convex under the natural restriction that the weight decreases with the singular value. However, we show that the proposed non-convex penalized regression method has a global optimal solution obtained from an adaptively soft-thresholded singular value decomposition. This new reduced-rank estimator is computationally efficient, has continuous solution path and possesses better bias-variance property than its classical counterpart. The rank consistency and prediction/estimation performance bounds of the proposed estimator are established under high-dimensional asymptotic regime. Simulation studies and an application in genetics demonstrate that the proposed estimator has superior performance to several existing methods. The adaptive nuclear-norm penalization can also serve as a building block to study a broad class of singular value penalties.
\end{abstract}

%

\section{Introduction}

Given $n$ observations of the response $\B{y}_i\in \Re^q$ and predictor $\B{x}_i\in \Re^p$, we consider the multivariate linear regression model:
\begin{equation}
\B{Y} = \B{X}\B{C} +\B{E}, \label{model1}
\end{equation}
where $\B{Y} = (\B{y}_{1},...,\B{y}_n)^\T$, $\B{X}=(\B{x}_1,...,\B{x}_n)^\T$, $\B{C}$ is a $p\times q$ coefficient matrix, and $\B{E}=(\B{e}_1,...,\B{e}_n)^\T$ is a random $n\times q$ matrix of independently and identically distributed random errors with mean zero and variance $\sigma^2$.

We are interested in the scenario when both the number of predictors $p$ and the number of responses $q$ may depend on and even exceed the sample size $n$. Such high-dimensional regression problems are increasingly encountered in quantitative investigations. It is well known that ordinary least squares (OLS) estimation is equivalent to separately regressing each response on the set of predictors, which, however, ignores the dependence structure of the multivariate response and may be infeasible in high-dimensional settings. The curse of dimensionality may be mitigated by assuming the coefficient matrix $\B{C}$ admit some low-dimensional structure and employing the regularization/penalization approach for model estimation. For example, for Gaussian data, it is appropriate to conduct model estimation by penalized least squares (PLS):
\begin{align}
\frac{1}{2}\mathcal{J}(\B{C})+ \mathcal{P}_\lambda(\B{C}),\label{criterion1}
\end{align}
where $\mathcal{J}(\B{C})=\|\B{Y}-\B{X}\B{C}\|_F^2$ is the sum of squared error with $\|\cdot\|_F$ denoting the Frobenius norm, $\mathcal{P}_{\lambda}(\cdot)$ is some penalty function measuring the  ``size'' (complexity)  of the enclosed matrix, and $\lambda$ is a non-negative tuning parameter controlling the degree of penalization.

Within this general framework, an important model is reduced-rank regression (RRR) \citep{anderson1951,anderson1999, anderson2002, izenman1975, reinsel1998}, in which dimension reduction is achieved by assuming that the coefficient matrix $\B{C}$ is of low rank, i.e., its rank  $r(\B{C})=r^*<\min(p,q)$. The classical RRR literature mainly focuses on small $p$ cases and maximum likelihood estimation.
\citet{bunea2011} proposed the rank selection criterion (RSC) for high dimensional settings, and revealed that rank constrained estimation can be viewed as a PLS method (\ref{criterion1}) with the penalty proportional to the rank of the coefficient matrix. This $l_0$-type penalty can be alternatively cast as a penalty in terms of the number of non-zero singular values of $\B{C}$, i.e., $\mathcal{P}_\lambda(\B{C})=\lambda r(\B{C})=\lambda\sum_k I(d_k(\B{C})\neq 0)$ where $I(\cdot)$ is the indicator function, which results in an estimator obtained by hard-thresholded singular value decomposition (SVD), see Section \ref{sec2}. \citet{yuan2007} proposed a nuclear-norm penalized least squares estimator (NNP), in which the nuclear norm penalty is defined as $\mathcal{P}_\lambda(\B{C})=\lambda\|\B{C}\|_*=\lambda\sum_k d_k(\B{C})$, where $\|\cdot\|_*$ denotes the nuclear norm or the sum of the singular values of the enclosed matrix. This $l_1$-type penalty encourages sparsity among the singular values and achieves simultaneous rank reduction and shrinkage coefficient estimation \citep{negahban2011,bunea2011,Lu2012}. \citet{Rohde2011} investigated the theoretical properties of the Schatten-$q$ quasi-norm penalty, which is defined as $\mathcal{P}_\lambda(\B{C})=\lambda\sum_{k=1}^{r} d^q_k(\B{C})$ for $0<q\leq 1$, and nonasymptotic bounds of prediction risk were obtained. Several other methods or theoretical development related to reduced-rank regression exist, see, e.g., \citet{aldrin2000}, \citet{negahban2011}, \citet{mukh2011}, and \citet{chen2011jrssb}. The reduced-rank metrology has connections with many popular tools including principal component analysis and canonical correlation analysis, and it is extensively studied in matrix completion problems \citep{candes2009,candes2011,koltch2011}.

It is evident that the aforementioned reduced-rank approaches are closely related to the SVD method in matrix approximation \citep{eckart1936,reinsel1998}. It is also intriguing that the rank and nuclear norm penalized methods can be viewed as $l_0$ and $l_1$ penalized methods in the SVD domain, respectively. (In fact, the $l_2$ norm of the singular values is equivalent to a ridge penalty.) Motivated by these connections, we propose the adaptive nuclear norm regularization method, in which the adaptive nuclear norm (ANN) of a matrix $\B{C}$ is defined as a weighted sum of its singular values; see \citet{xu2009} for a similar idea related to reweighted penalization in the context of matrix completion. Clearly, the key is to close the gap between the nonsmooth $l_0$ rank penalty and the $l_1$ nuclear penalty while keeping the computation stable and efficient for high dimensional data. We show that the convexity of ANN depends on the ordering of the weights. Particularly, the ANN turns out to be non-convex for the case that the weight decreases with the singular value, a condition needed for a meaningful regularization, see Section \ref{sec2} for further discussion. Despite the non-convexity, we are able to characterize the explicit global optimal solution for the ANN penalized matrix approximation problems.

Based on ANN, we develop a new method of simultaneous dimension reduction and coefficient estimation for the general high-dimensional multivariate regression. Our proposal is based on two main ideas. Firstly, by penalizing the singular values adaptively, our method builds a bridge between the RSC and NNP methods, and it may be viewed as analogous to the adaptive Lasso \citep{tib1996,zou2006} developed for univariate regression. Secondly, the ANN penalty is applied to $\B{XC}$ rather than $\B{C}$ \citep{koltch2011}; although the criterion remains non-convex, this setup allows the reduced-rank estimation to be solved explicitly and efficiently. Comparing to the computationally intensive NNP method which tends to overestimate the rank, the proposed ANN method may improve rank determination with the aid of some well-chosen adaptive weights. Comparing to the RSC method, the smooth ANN penalty results in a continuous solution path and allows more flexible bias and variance tradeoff in model fitting. The rank consistency and prediction/estimation performance bounds of the proposed estimator are established under high-dimensional asymptotic regime. We discuss the incorporation of an extra $l_2$ penalization for improving reduced-rank estimation. Empirical studies demonstrate that the proposed methods enjoy superior performance in both prediction and rank estimation as compared to several existing methods.

\section{Adaptive nuclear norm penalty}\label{sec2}

To study the general properties of the adaptive nuclear norm, we consider the low-rank matrix approximation problem, $\B{Y} = \B{C} + \B{E}$, which is a special case of model (\ref{model1}) when $\B{X}$ becomes an identity matrix and $n=p$. In many applications, given the noisy data matrix $\B{Y}$, it is of  interest to seek its low-rank approximation for denoising, which can be achieved by various methods, e.g., rank penalization or nuclear norm penalization. These methods are closely related to the SVD method in low-rank matrix approximation, which motivated our study.

Consider the SVD of $\B{Y} \in \Re^{n\times q}$,
\begin{equation}
\B{Y}=\B{UDV}^\T, \qquad \B{D}=\Diag\{\B{d}(\B{C})\}=\Diag\{\B{d}\},\label{svdQ}
\end{equation}
where $\B{U}$ and $\B{V}$ are respectively $p\times h$ and $q \times h$ orthonormal matrices with $h=\min(p,q)$, and the vector of singular values $\B{d}=(d_1,...,d_h)$ consists of non-increasing non-negative singular values of $\B{Y}$. For any $\lambda \geq 0$, define the hard SVD-thresholding operator (HSVT)
\begin{equation}
    \mathcal{H}_\lambda(\B{Y})=\B{U}\mathcal{H}_\lambda(\B{D})\B{V}^\T, \qquad \mathcal{H}_\lambda(\B{D})=\Diag\{d_iI(d_i>\lambda),i=1,...,h\},\label{hth}
\end{equation}
where $I(\cdot)$ is the indicator function, and the soft SVD-thresholding operator (SSVT)
\begin{equation}
    \mathcal{S}_\lambda(\B{Y})=\B{U}\mathcal{S}_\lambda(\B{D})\B{V}^\T, \qquad \mathcal{S}_\lambda(\B{D})=\Diag\{(d_i-\lambda)_{+},i=1,...,h\},\label{sth}
\end{equation}
where $x_+$ is the non-negative part of $x$, namely, $x_+ = \max(0, x)$. It is well-known that in the matrix approximation problems ($\B{X}=\B{I}$), an HSVT estimator solves (\ref{criterion1}) with the $l_0$-type rank penalization \citep{eckart1936}, while a SSVT estimator solves (\ref{criterion1}) with the $l_1$-type nuclear-norm penalization \citep{cai2008}; these results are summarized in the proposition below.

\begin{proposition}\label{prop1}
For any $\lambda \geq 0$ and $\B{Y}\in \Re^{n\times q}$, the HSVT operator $\mathcal{H}_\lambda(\B{Y})$ defined by  (\ref{hth}) can be characterized as:
$
\mathcal{H}_\lambda(\B{Y})=\arg\min_\B{C}\{\|\B{Y}-\B{C}\|^2_F+\lambda^2 \M{r}(\B{C})\}
$,  and the SSVT operator $\mathcal{S}_\lambda(\B{Y})$ in (\ref{sth}) is similarly characterized as:
$
\mathcal{S}_\lambda(\B{Y})=\arg\min_\B{C}\{\frac{1}{2}\|\B{Y}-\B{C}\|^2_F+\lambda \|\B{C}\|_*\}
$.
\end{proposition}

The hard-thresholding operator $\mathcal{H}_\lambda(\B{Y})$ eliminates any singular values below a threshold value $\lambda$, while the soft-thresholding operator $\mathcal{S}_\lambda(\B{Y})$ shrinks all the singular values by the same amount $\lambda$ towards zero. These two SVD operators are natural extensions of the hard/soft-thresholding rules for scalars and vectors \citep{donoho1995,cai2008}. Generally, estimators based on hard-thresholding often have small bias but may suffer from large variance; in contrast, soft-thresholding reduces variance by introducing extra bias in the estimators, which may be preferable in cases when data are noisy and highly correlated \citep{donoho1995}.

The connections between different thresholding rules and penalty terms motivated us to consider the adaptive nuclear norm penalization (ANN). The main idea is to build a bridge between the $l_0$ and $l_1$ penalties or the HSVT and SSVT rules so as to fine-tune the bias-variance tradeoff in the SVD domain. We define the adaptive nuclear norm of a matrix $\B{C}\in \Re^{p\times q}$:
\begin{equation}
f(\B{C})=\|\B{C}\|_{*\B{w}}=\sum_{i=1}^{h} w_i d_i(\B{C}),\label{def:ann}
\end{equation}
where $h=\min(p,q)$, $d_i(\cdot)$ is the $i$th largest singular value of the enclosed matrix, and the $w_i$s are the non-negative weights.

Since the nuclear norm is convex and is a matrix norm, a natural question arises as to whether or not its weighted extension ANN preserves the convexity, which is the case for lasso and adaptive lasso penalties in the vector case \citep{zou2006}. However, our analysis shows that the convexity of the ANN depends on the ordering of the non-negative weights. The following theorem gives a necessary and sufficient condition of its convexity.

\begin{theorem}\label{thm:convexity}
For any matrix $\B{C}\in \Re^{p\times q}$ ($n=p$, $h=\min(p,q)$), let $f(\B{C})=\|\B{C}\|_{*\B{w}}$ defined in (\ref{def:ann}). Then $f(\cdot)$ is convex if and only if $w_1\geq w_2 \geq\cdots\geq w_h\geq 0$.
\end{theorem}

Hence, for the ANN to be a convex function, the weight must increase with the singular value, i.e., they are co-monotone. However, to use ANN for penalized estimation, the opposite is desirable, i.e., we would and shall henceforth impose the following order constraint:
\begin{equation}
0\leq w_1 \leq \cdots \leq w_h, \label{order-constraint}
\end{equation}
in order for larger singular values to receive lesser penalty to help reducing the bias and smaller singular values to receive heavier penalty to help promoting sparsity. Here is an example showing that the ANN is neither convex nor concave under constraint (\ref{order-constraint}). Consider $n=p = q = 2$, and
$$
\B{C}_1 = \begin{pmatrix} 2 & 0 \\ 0 & 1 \end{pmatrix}, \qquad\B{C}_2 = \begin{pmatrix}1&0\\0&2\end{pmatrix}.
$$
Let $w_1=1$ and $w_2 = 2$. It can be verified that $f(\B{C}_1) = f(\B{C}_2) =f(-\B{C}_2)=4$, while $f((\B{C}_1+\B{C}_2)/2) = 4.5 > (f(\B{C}_1)+f(\B{C}_2))/2$; also, $f((\B{C}_1-\B{C}_2)/2)=1.5<(f(\B{C}_1)+f(-\B{C}_2))/2$.

The non-convexity of the ANN arises from the constraint (\ref{order-constraint}) that the weight decreases with the singular value. In fact, the ANN is then \emph{no longer} a matrix norm. However, we are able to explicitly solve and characterize the global solution of the ANN criterion as follows.

\begin{theorem}\label{thm:optimality}
For any $\lambda \geq 0$, $\B{Y}\in \Re^{n\times q}$ and $0\leq w_1 \leq \cdots \leq w_{h}$ ($n=p$, $h=\min(n,q)$), a global optimal solution to the optimization problem
\begin{equation}
\min_\B{C}f(\B{C}):=\left\{\frac{1}{2}\|\B{Y}-\B{C}\|^2_F+\lambda \sum_{i=1}^{h} w_i  d_i ( \B{C})\right\}\label{eq:DC}.
\end{equation}
is given by the adaptive SVD soft-thresholding (ASVT) operator   $\hat{\B{C}}:=\mathcal{S}_{\lambda\B{w}}(\B{Y})$,
\begin{equation}
    \mathcal{S}_{\lambda\B{w}}(\B{Y})=\B{U}\mathcal{S}_{\lambda\B{w}}(\B{D})\B{V}^\T, \qquad \mathcal{S}_{\lambda\B{w}}(\B{D})=\Diag\{(d_i-\lambda w_i)_{+},i=1,...,h\},\label{ast}
\end{equation}
Further, if $\B{Y}$ has a unique SVD, $\hat{\B{C}}$ is the unique optimal solution.
\end{theorem}

The fact that a closed-form global minimizer can be found for the non-convex ANN problem is not immediately clear and rather surprising. The result stems from the von Neumann's trace inequality \citep{Mirsky1976} and the properties of SVD, see the Appendix for details. Following \citet{zou2006}, the weights can be set as some power of the singular values of the data matrix, i.e., $\B{w}=\{\B{d}(\B{Y})\}^{-\gamma}$, where $\gamma\geq0$ is a prespecified constant. In this way, the order constraint (\ref{order-constraint}) is automatically satisfied. A more general way of constructing the weights and its relation to other penalty forms will be discussed in Section \ref{sec:discussion}.

\section{Adaptive nuclear norm penalization in regression}\label{sec3}

We now consider the general problem of estimating the coefficient matrix $\B{C}$, which is possibly of low-rank, in the multivariate linear regression model (\ref{model1}). Below, let $\B{P}=\B{X}(\B{X}^\T\B{X})^{-}\B{X}^\T$ be the projection matrix onto the column space of $\B{X}$ and $\tD{\B{C}}=(\B{X}^\T\B{X})^{-}\B{X}^\T\B{Y}$ the LS estimator of $\B{C}$, where $(\B{X}^\T\B{X})^{-}$ denotes the Moore-Penrose inverse of the enclosed Gram matrix. Unless otherwise noted, the singular values are always placed in non-increasing order.

\subsection{Rank and nuclear norm penalized regression methods}

The fundamental results in Theorem \ref{thm:optimality} about rank and nuclear norm penalization for matrix approximation can be readily extended to the general regression setting. Consider first the rank penalized least squares criterion \citep{bunea2011},
\begin{equation}
\frac{1}{2}\|\B{Y}-\B{X}\B{C}\|^2_F+\lambda r(\B{C}).\label{rrobj}
\end{equation}
Based on Proposition \ref{prop1}, it can be easily shown that the minimizer of (\ref{rrobj}), denoted as $\tD{\B{C}}^{(\lambda)}$, can be obtained by hard-thresholding the SVD of $\B{X}\tD{\B{C}}$. Let $\tD{\B{V}}\tD{\B{D}}^2\tD{\B{V}}^\T$ be the eigenvalue decomposition of $\B{Y}^\T\B{P}\B{Y}=(\B{X}\tD{\B{C}})^\T\B{X}\tD{\B{C}}$. The SVD of $\B{X}\tD{\B{C}}$ is then given by $\tD{\B{U}}\tD{\B{D}}\tD{\B{V}}^\T$, where $\tD{\B{U}}=\B{PY}\tD{\B{V}}\tD{\B{D}}^{-1}=\B{X}\tD{\B{C}}\tD{\B{V}}\tD{\B{D}}^{-1}$. Therefore,
\begin{align}
&\B{X}\tD{\B{C}}^{(\lambda)} = \mathcal{H}_{\sqrt{2\lambda}}(\B{X}\tD{\B{C}})= \B{X}\tD{\B{C}}\tD{\B{V}}\tD{\B{D}}^{-1}\mathcal{H}_{\sqrt{2\lambda}}(\tD{\B{D}})\tD{\B{V}}^\T,\qquad\tD{\B{C}}^{(\lambda)}= \tD{\B{C}}\tD{\B{V}}\tD{\B{D}}^{-1}\mathcal{H}_{\sqrt{2\lambda}}(\tD{\B{D}})\tD{\B{V}}^\T.\label{RSC}
\end{align}
This rank selection criterion (RSC) proposed by \citet{bunea2011} is valid in high-dimensional settings and hence extends the classical rank-constrained RRR approach \citep{reinsel1998}. In fact, the set of rank-constrained estimators which minimize $\|\B{Y}-\B{XC}\|_F^2$ subject to $r(\B{C})=r$ ($r=1,...,\min(p,q)$), spans the solution path of (\ref{rrobj}).

The nuclear-norm penalized least squares criterion (NNP) \citep{yuan2007}
\begin{equation}
\frac{1}{2}\|\B{Y}-\B{X}\B{C}\|^2_F+\lambda \|\B{C}\|_*,\label{nnobj}
\end{equation}
does not have an explicit solution in general, and can be computationally intensive for large-scale data. Extensive research has been devoted to its optimization problem, e.g., \citet{cai2008}, \citet{toh2009}, etc. One popular algorithm is to iteratively conducting a majorization step of the objective function and a minimization step using soft SVD-thresholding \citet{cai2008}.

The performance of the RSC and NNP estimators is related to the bias-variance trade-off phenomenon discussed in Section \ref{sec2}. The NNP may be more accurate than the RSC when the correlation among predictors is high or the signal to noise ratio (SNR) is low, while RSC may perform better when the correlation is moderate and the SNR is not too low; see Section \ref{sec5} for details. A drawback of NNP is that it is computationally intensive and is generally not as parsimonious as RSC in rank determination. These motivated our study of the ANN for building a continuum of estimators between the RSC and NNP estimators.

\subsection{Adaptive nuclear norm penalized regression method}

Predictive accuracy and computation efficiency are both pivotal in high dimensional regression problems. Motivated by criteria (\ref{rrobj}) and (\ref{nnobj}) and their connections with SVD, we propose to estimate $\B{C}$ by minimizing the non-convex PLS criterion
\begin{equation}
\frac{1}{2}\|\B{Y}-\B{X}\B{C}\|^2_F+\lambda \sum_{i=1}^{h} w_i d_i(\B{X}\B{C}),\label{wnnobj}
\end{equation}
where $h=\min(p,q)$ and the weights $\{w_i\}$ are required to be non-negative and in non-decreasing order. In practice, a foremost task of using ANN is setting proper adaptive weights. Following \citet{zou2006}, a natural way to construct the weights is based on the LS solution:
\begin{align}
\B{w}=\{\B{d}(\B{PY})\}^{-\gamma}=\tD{\B{d}}^{-\gamma},\label{eq:weights}
\end{align}
where $\B{PY}$ is the projection of $\B{Y}$ onto the column space of $\B{X}$ and $\gamma$ a non-negative constant.

The proposed criterion (\ref{wnnobj}) is built on two main ideas. Firstly, the criterion directly focuses on prediction matrix approximation and encourages sparsity among the singular values of $\B{X}\B{C}$ rather than those of $\B{C}$, which may yield low-rank solutions for $\B{XC}$ and hence for $\B{C}$. A prominent advantage of this setup is that the problem can then be solved explicitly and efficiently. Secondly, the adaptive penalization of the singular values allows flexible bias-variance tradeoff: large singular values receive small amount of penalization to control the possible bias, and small singular values receive large amount of penalization to induce sparsity and hence reduce the rank. The following Corollary shows that this criterion leads to an explicit ANN estimator.

\begin{corollary}
A minimizer of (\ref{wnnobj}), denoted as $\hat{\B{C}}^{(\lambda\B{w})}$, is obtained via adaptively soft-thresholding the SVD of $\B{X}\tD{\B{C}}$ where $\tD{\B{C}}$ is the LS estimator of $\B{C}$, i.e.,
\begin{align}
\B{X}\hat{\B{C}}^{(\lambda \B{w})} = \mathcal{S}_{\lambda\B{w}}(\B{X}\tD{\B{C}})=\tD{\B{U}}\mathcal{S}_{\lambda\B{w}}(\tD{\B{D}})\tD{\B{V}}^\T,\qquad\hat{\B{C}}^{(\lambda \B{w})}= \tD{\B{C}}\tD{\B{V}}\tD{\B{D}}^{-1}\mathcal{S}_{\lambda\B{w}}(\tD{\B{D}})\tD{\B{V}}^\T.\label{astest}
\end{align}
where $\tD{\B{U}}\tD{\B{D}}\tD{\B{V}}^\T$ is the SVD of $\B{X}\tD{\B{C}}$ as defined in the previous section.
\end{corollary}

By Pythagoras' theorem, minimizing the criterion (\ref{wnnobj}) is equivalent to
minimizing $
\{1/2\|\B{X}\tD{\B{C}}-\B{X}\B{C}\|^2_F+\lambda \sum_{i} w_i d_i(\B{X}\B{C})\}$ with respect to $\B{C}$, where $\tD{\B{C}}$ is the OLS estimator. The above result then directly follows from Theorem \ref{thm:optimality}. The proposed method first projects $Y$ onto the column space of $X$, i.e., $PY=X\tilde{C}$, and the ANN estimator is then obtained as a low-rank approximation of $PY$ via soft SVD-thresolding; the thresholding level is adaptive and can be data-driven: the smaller a singular value, the larger its thresholding level. Therefore, the estimated rank of an ANN estimator corresponds to the smallest singular value of $PY$ that exceeds its thresholding level, i.e., $\hat{r}=\max\{r: d_r(PY)> \lambda w_r\}$. For the choice of the weights (\ref{eq:weights}), i.e., the estimated rank is given by
\begin{align}
\hat{r}=\max\{r: d_r(\B{PY})> \lambda^{\frac{1}{\gamma+1}}\},\label{est:rank}
\end{align}
and the plausible range of the tuning parameter is $\lambda \in [0,\tD{d}_1^{\gamma+1}]$, with $\lambda=0$ corresponding to the LS solution and $\lambda=\tD{d}_1^{\gamma+1}$ the null solution.

The ANN estimator and the RSC estimator only differ in their singular values but the difference can be consequential. While the solution path of RSC is discontinuous and the number of possible solutions equals to the maximum rank, the ANN criterion offers more flexibility in that the resulting solution path is continuous and guided by the data-driven weights. The ANN and RSC are based on the same one-time SVD operation and thus they have similar computation complexity and can both be easily implemented and efficiently computed, in contrast to the computationally intensive NNP method.


For any fixed $\lambda>0$, the ANN estimator $\hat{\B{C}}^{(\lambda \B{w})}$ can be computed by (\ref{astest}). (The same SVD operation can be used to compute the RSC solutions.) To choose an optimal $\lambda$ and hence an optimal ANN solution, we use the $K$-fold cross validation (CV) method, based on the predictive performance of the models \citep{stone1974}. For the numerical studies reported below, we first compute the solutions over a grid of 100 $\lambda$ values equally spaced on the log scale and select the best $\lambda$ value by CV;  subsequently we refine the selection process around the chosen $\lambda$ value with another finer grid of $100$ $\lambda$ values.

\section{Rank consistency and error bounds}\label{sec4}

We study the rank estimation and prediction properties of the proposed ANN estimator. Our theoretical analysis is built on the framework developed by \citet{bunea2011}, as RSC and ANN are closely connected. We mainly focus on the random weights constructed in (\ref{eq:weights}), in line with the adaptive Lasso method \citep{zou2006} developed for the univariate (multiple) regression. Similar results are obtained for any prespecified sequence of weights satisfying certain order restriction and boundedness requirements. All the proofs are given in the Appendix.

The rank of the coefficient matrix $\B{C}$ can be viewed as the number of effective combination of predictors linked to the responses. Rank determination is always a foremost task of reduced-rank estimation. The quality of rank estimator, defined by (\ref{est:rank}), clearly depends on the signal to noise ratio. Following \citet{bunea2011}, we shall use the smallest non-zero singular value of $\B{XC}$, i.e., $d_{r^*}(\B{XC})$, to measure the signal strength, and use the largest singular value of the projected noise matrix $\B{PE}$, i.e., $d_1(\B{PE})$, to measure that of the noise. Intuitively, if $d_1(\B{PE})$ is much larger than the size of the signal, some signal could be deeply masked by the noise and lost during the thresholding procedure; as such, $\hat{r}$ may be much smaller than the true rank. The lemma below characterizes the ``limit'' or the true target of $\hat{r}$ and its relationship with the noise level.

\begin{lemma}\label{lemma:rank1}
Suppose that there exists an index $s\leq r^*$ such that
$d_s(\B{XC})> (1+\delta)\lambda^{1/(\gamma+1)}$ and $d_{s+1}(\B{XC}) \leq (1-\delta)\lambda^{1/(\gamma+1)}$ for some $\delta \in (0,1]$. Then
$\M{P}(\hat{r}=s)\geq 1- \M{P}(d_1(\B{PE})\geq \delta\lambda^{1/(\gamma+1)})$, where $\B{P}$ is the projection matrix onto the column space of $\B{X}$, $\B{E}$ is the error matrix in model (\ref{model1}), and $\gamma$ is the power parameter in the adaptive weights (\ref{eq:weights}).
\end{lemma}

This result establishes the relationship between the estimated rank, the signal level, the noise level and the adaptive weights. To achieve consistent rank estimation, we consider the following assumptions:

\begin{assumption}
    The error matrix $\B{E}$ has independent $N(0, \sigma^2)$ entries.
\end{assumption}
\begin{assumption}
    For any $\theta>0$, assume $\lambda=\{(1+\theta)\sigma(\sqrt{r_x}+\sqrt{q})/\delta\}^{\gamma+1}$ with $\delta$ defined in Lemma \ref{lemma:rank1}, and assume $d_{r^*}(\B{XC})>2\lambda^{1/(\gamma+1)}$.
\end{assumption}

Assumption 1 is about the error structure, which ensures that the noise level $d_1(\B{PE})$ is of order $\sqrt{r_x}+\sqrt{q}$, see Lemma \ref{lemma:rank2} \citep{bunea2011}. Assumption 2 concerns the signal strength relative to the noise level and the appropriate rate of the tuning parameter. 

\begin{theorem}\label{th:rank}
Suppose Assumptions 1--2 hold. Let $r^*=r(\B{C})$ be the true rank, $r_x=r(\B{X})$ be the rank of $\B{X}$, and $\hat{\B{r}}$ be the estimated rank defined by  (\ref{est:rank}). Then $\M{P}(\hat{r}=r^*) \rightarrow 1$ as $r_x+q \rightarrow \infty$.
\end{theorem}

Theorem \ref{th:rank} shows that the ANN estimator is able to identify the correct rank with probability tending to 1 as $r_x+q$ goes to infinity. Similar to \citet{bunea2011}, the consistency results can be extended to the case of sub-Gaussian errors and can also be easily adapted to the case when $r_x + q$ is bounded and the sample size $n$ goes to infinity. Therefore, the rank consistency of the proposed ANN estimator is valid for both classical and high-dimensional asymptotic regimes.

Our main results about the prediction performance of the proposed ANN estimator are presented in Theorem \ref{bound:th2} below. For simplify, we write $\hat{\B{C}}$ for $\hat{\B{C}}^{(\lambda\B{w})}$.

\begin{theorem}\label{bound:th2}
Suppose Assumptions 1--2 hold. Let $c=d_{1}(\B{XC})/d_{r^*}(\B{XC})\geq 1$. Then
\begin{equation}
\|\B{X}\hat{\B{C}}-\B{XC}\|_F^2 \leq \frac{1+a}{1-a}\|\B{X}\B{B}-\B{XC}\|_F^2
+\frac{1}{a(1-a)}\left\{\sqrt{2}\delta+2(2-\delta)^{-\gamma}-(2c+\delta)^{-\gamma}\right\}^2\lambda^{\frac{2}{\gamma+1}}r^*,\notag
\end{equation}
with probability greater than $1-\exp(-\theta^2(r_x+q)/2)$, for any $0<a<1$ and any $p\times q$ matrix $\B{B}$ with $r(\B{B})\leq r^*$. Moreover, taking $\B{B}=\B{C}$ and $a=1/2$ yields
\begin{align*}
\|\B{X}\hat{\B{C}}-\B{XC}\|_F^2 & \leq 4\{\sqrt{2}\delta+2(2-\delta)^{-\gamma}-(2c+\delta)^{-\gamma}\}^2\lambda^{\frac{2}{\gamma+1}}r^*\\
& = 4\{\sqrt{2}+2(2-\delta)^{-\gamma}/\delta-(2c+\delta)^{-\gamma}/\delta\}^2(1+\theta)^2\sigma^2(\sqrt{r_x}+\sqrt{q})^2r^*
\end{align*}
with probability greater than $1-\exp(-\theta^2(r_x+q)/2)$.
\end{theorem}

The above established bound shows that the prediction error is bounded by $d_1^2(\B{PE}) r^*$ up to some constant with probability $1-\exp(-\theta^2(r_x+q)/2)$, i.e., the smaller the error size or the true rank, the smaller the prediction error. The bound is valid for any $\B{X}$ and $\B{C}$. The estimation error bound of $\hat{\B{C}}$ can also be readily derived from Theorem \ref{bound:th2}, e.g., if $d_{r_x}(\B{X})\geq \rho>0$ for some constant $\rho$, then under Assumptions 1--2,
$\|\hat{\B{C}}-\B{C}\|_F^2  \leq 4\rho^{-2}\{\sqrt{2}\delta+2(2-\delta)^{-\gamma}-(2c+\delta)^{-\gamma}\}^2\lambda^{\frac{2}{\gamma+1}}r^*$.

The rank consistency and prediction bound can be similarly established for any prespecified sequence of weights satisfying \begin{align}
0 \leq w_1\leq\cdots \leq w_{\bar{r}}, w_{r^*}\leq M \leq w_{r^*+1},\label{fixedweights1}
\end{align}
where $0 < M <\infty$, $w_{r^*+1}>0$ and $\bar{r}=\min(r_x,q)$; the index $s$ in Lemma 1, the requirements on the tuning sequence and the signal level in Assumption 2 shall be modified accordingly,
\begin{align}
& d_s(\B{XC})> (1+\delta)\lambda w_s, d_{s+1}(\B{XC}) < (1-\delta)\lambda w_{s+1} \M{ for some } \delta \in (0,1],\label{fixedweights2}\\
& \lambda=(1+\theta)\sigma(\sqrt{r_x}+\sqrt{q})(\delta M)^{-1}, d_{r^*}(\B{XC})>2\lambda M.\label{fixedweights3}
\end{align}
\begin{corollary}\label{corollary:2}
Suppose that Assumption 1 and (\ref{fixedweights1})--(\ref{fixedweights3}) hold. Then
\begin{align*}
&(1)\,\M{P}(\hat{r}=r^*) \rightarrow 1 \M{ as } r_x+q \rightarrow \infty;\\
&(2)\, \|\B{X}\hat{\B{C}}-\B{XC}\|_F^2 \leq \frac{1+a}{1-a}\|\B{X}\B{B}-\B{XC}\|_F^2
+\frac{1}{a(1-a)}\{(2+\sqrt{2}\delta)M-w_1\}^2\lambda^2r^* \\
&\qquad\M{ for any } 0 < a < 1 \M{ and } B \M{ with } r(B)\leq r^*;\\
&(3) \, \|\B{X}\hat{\B{C}}-\B{XC}\|_F^2 \leq 4(\sqrt{2}+2/\delta-w_1/(M\delta))^2(1+\theta)^2\sigma^2(\sqrt{r_x}+\sqrt{q})^2r^*\\
&\qquad\M{ with probability greater than } 1-\exp(-\theta^2(r_x+q)/2).
\end{align*}
\end{corollary}
The proof is similar to that of Theorems \ref{th:rank} and \ref{bound:th2} and hence is omitted.

The error bounds of the ANN estimator established in Theorem \ref{bound:th2} and Corollary \ref{corollary:2} are comparable to those of the RSC and NNP estimators \citep{bunea2011,Rohde2011}. The rate of convergence is $(r_x+q)r^*$, which is the optimal minimax rate for rank sparsity under suitable regularity conditions \citep{Rohde2011,bunea2012joint}. However, the bounds for NNP was obtained with some extra restrictions on the design matrix, and its tuning sequence that achieves the smallest mean squared error (MSE) usually does not lead to correct rank recovery \citep{bunea2011}. While both RSC and ANN are able to achieve correct rank recovery and minimal MSE simultaneously, the latter possess continuous solution path with data-driven tuning which may lead to improved empirical performance.

\section{Robustification of the reduced-rank estimation}\label{sec:robust}

As suggested by a referee and motivated by \citet{mukh2011}, we discuss the robustification of the reduced-rank methods by incorporating extra $l_2$ penalty in the penalized criteria.

\citet{mukh2011} proposed the robust reduced-rank ridge (RoRR) method which performs $l_2$ penalized ridge regression under rank constraint. The shrinkage estimation induced by the $l_2$ penalty makes the reduced rank estimation robust and especially suitable when the predictors are highly correlated. The method can be viewed as the following PLS criterion that was also mentioned in \citet{bunea2011},
\begin{align}
\frac{1}{2}\|Y-XC\|^2_F+\lambda_1 r(C) + \frac{1}{2}\lambda_2 \sum_{i=1}^{h} d_i^2(C), \label{RoRRR}
\end{align}
where $h=\min(p,q)$, $\sum_i d_i^2(C)=\M{tr}(C^\T C)$, and $\lambda_1$ and $\lambda_2$ are tuning parameters. The problem can be solved via data augmentation. Specifically, let
$$
 Y^*= \begin{pmatrix} Y \\ 0_{p\times q} \end{pmatrix}, \qquad X^* = \begin{pmatrix}X\\ \sqrt{\lambda_2}I_{p\times p}\end{pmatrix},
$$
then (\ref{RoRRR}) can be written as an RSC criterion $1/2\|Y^*-X^*C\|_F^2+\lambda_1 r(C)$, whose solution is given by (\ref{RSC}); see \citet{mukh2011}.

Similarly, the proposed ANN method can also be robustified by incorporating a ridge penalty term. Similar to (\ref{wnnobj}), for efficient computation, we impose an $l_2$ penalty on $XC$ rather than $C$,
\begin{align}
\frac{1}{2}\|Y-XC\|_F^2+\lambda_1 \sum_{i=1}^{h}w_i d_i(XC) + \frac{1}{2}\lambda_2\sum_{i=1}^{h} d_i^2(XC).\label{RoANN}
\end{align}
Interestingly, this robustified ANN criterion (RoANN) is  analogous to the adaptive elastic net criterion \citep{zou2005} in univariate regression. (The case of imposing one or both penalties on $C$ directly is more complex and will be pursued elsewhere.) It can be easily verified that, for fixed tuning parameters, the objective function (\ref{RoANN}) is minimized at
$$
\hat{C}=\frac{1}{1+\lambda_2}\hat{C}^{(\lambda_1 w)}
$$
where $\hat{C}^{(\lambda_1 w)}$ denotes the ANN estimator in the absence of the $l_2$ penalty. Indeed, the extra $l_2$ penalty induces overall shrinkage of the ANN estimator.

For each fixed $\lambda_2$, the RoRR method requires inverting a $p\times p$ matrix of the form $(\B{X}^T\B{X}+\lambda_2\B{I})$ and performing an SVD of a $q\times q$ matrix. When $p$ is much bigger than $n$, the Woodbury matrix identity is useful in speeding up computation \citep{Hager1989}, i.e., $(\B{X}^T\B{X}+\lambda_2\B{I})^{-1}=1/\lambda_2\B{I}-1/\lambda_2^2\B{X}^T(\B{I}+1/\lambda_2\B{XX}^T)^{-1}\B{X}$. On the other hand, the RoANN method only requires one-time matrix inversion and SVD operation for obtaining the whole solution path, thereby saving computation. We shall examine the effects of an additional $l_2$ penalization by simulation in Section \ref{sec5}. Due to space limit, relevant theoretical analysis will be pursued elsewhere.

\section{Empirical studies}\label{sec5}

\subsection{Simulation}\label{sec5:1}

We compare the prediction, estimation and rank determination performances of the NNP estimator proposed by \citet{yuan2007}, the RSC estimator proposed by \citet{bunea2011}, the RoRR estimator proposed by \citet{mukh2011}, and our proposed ANN and RoANN estimator. In ANN estimation the adaptive weights are constructed as (\ref{eq:weights}) with $\gamma=0,1,2$ and we denote the resulting estimator as $\M{ANN}_\gamma$ ($\gamma=0$ means unweighted ANN). In the numerical results reported below, we use the accelerated proximal gradient algorithm implemented in Matlab by \citet{toh2009} for NNP estimation. R code for RoRR was provided by their original authors \citep{mukh2011}, and we modified their code to make use of the Woodbury matrix identity \citep{Hager1989} for saving computation. We have also implemented all the other methods in R \citep{r2008}. All computation was done on linux machines with 3.4 GHz CPU and 8 GB RAM.

We consider the same simulation models as in \citet{bunea2011}. The coefficient matrix $\B{C}$ is constructed as $\B{C} = b\B{C}_0\B{C}_1^\T$, where $b>0$, $\B{C}_0 \in \Re^{p\times r^*}$, $\B{C}_1\in \Re^{q\times r^*}$ and all entries in $\B{C}_0$ and $\B{C}_1$ are i.i.d. $N(0,1)$. Two scenarios of model dimensions are considered, i.e., $p,q<n$ and $p,q>n$.
\begin{itemize}
\item Model I ($n=100$, $p=q=25$, $r^*=10$): The covariate matrix $\B{X}$ is constructed by generating its $n$ rows as i.i.d. samples from a multivariate normal distribution $\M{MVN}(\B{0},\Gamma)$, where $\Gamma=(\Gamma_{ij})_{p\times p}$ and $\Gamma_{ij}=\rho^{|i-j|}$ with some $0<\rho<1$.
\item Model II ($n=20$, $p=q=25$, $r^*=5$, $r_x=10$): The covariate matrix $\B{X}$ is generated as $\B{X}=\B{X}_0\Gamma^{1/2}$, where $\Gamma$ is defined as above, $\B{X}_0=\B{X}_1\B{X}_2$, $\B{X}_1\in \Re^{n\times r_x}$, $\B{X}_2\in \Re^{r_x\times p}$, and all entries of $\B{X}_1$, $\B{X}_2$ are i.i.d $N(0, 1)$.
\end{itemize}
The data matrix $\B{Y}$ is then generated by $\B{Y}=\B{X}\B{C}+\B{E}$, where the elements of $\B{E}$ are i.i.d. samples from $N(0,1)$. It can be seen that each simulated model is characterized by the following parameters : $n$ (sample size), $p$ (number of predictors), $q$ (number of responses), $r^*$ (rank of $\B{C}$), $r_x$ (rank of $\B{X}$), $\rho$ (design correlation), and $b$ (signal strength). The experiment was replicated 500 times for each parameter setting.

To alleviate the influence of the inaccuracy in empirical tuning parameter selection and to reveal the true potential of each penalized method for fair comparison, one way is to tune each method based on its prediction accuracy evaluated on a very large independently generated validation data set; this yields ``optimally tuned'' estimators denoted as $\M{NNP}^{(O)}$, $\M{RSC}^{(O)}$, $\M{ANN}_{\gamma}^{(O)}$, etc. The 10-fold cross validation method is also used with the actual data, which results in the estimators $\M{NNP}^{(C)}$, $\M{RSC}^{(C)}$, $\M{ANN}^{(C)}$, etc. (Although to save space, we only report the CV results of $\M{RSC}$ and $\M{ANN}_2$.) For each method, the model accuracy is measured by the average of the scaled mean-squared-errors (SMSE) from all runs, i.e., $\M{SMSE}=100\|\B{C}-\hat{\B{C}}\|_F^2/(pq)$ for estimation (Est), and $\M{SMSE}=100\|\B{X}\B{C}-\B{X}\hat{\B{C}}\|_F^2/(nq)$ for prediction (Pred). Their standard errors are also reported. To evaluate the rank determination performance, we report (1) the average of the estimated ranks from all runs, and (2) the percentage of correct rank identification. For each method, the average computation time per replication is reported.

\begin{table}
\begin{center}
{\caption{\label{table1} Comparison of estimation, prediction and rank determination performances of various reduced-rank estimators using Model I ($n=100,p=25,q=25,r^*=10$). The superscript $^{(O)}$ stands for optimal tuning, and $^{(C)}$ stands for cross validation. The estimation error (Est) and prediction error (Pred) are reported along with their standard errors in the parenthese. For rank estimation (Rank), the average of estimated rank and the percentage of correct rank identification are reported. The simulation is based on 500 replications, and the average running time of each replication is reported in seconds (Time).}}
\begin{tiny}
\begin{tabular}{ll|rrrrrrrrr}
\hline
       $b$ &        ERR &                                                                                       \multicolumn{ 9}{|c}{Method} \\

           &            & $\M{NNP}^I$ & $\M{RSC}^O$ & $\M{RSC}^C$ & $\M{ANN}^{O}_2$ & $\M{ANN}^{C}_2$ & $\M{ANN}^{O}_1$ & $\M{ANN}^{O}_0$ & $\M{RoRR}^O$ & $\M{RoANN}^{O}_2$ \\
\hline
           &            &                                                                                   \multicolumn{ 9}{|c}{$\rho=0.9$} \\
\hline
      0.05 &        Est & 1.57 (0.2) &  3.20 (0.6) & 3.23 (0.7) & 2.61 (0.4) & 2.67 (0.5) & 2.48 (0.4) & 2.46 (0.3) &  1.60 (0.2) & 2.48 (0.4) \\

           &       Pred & 7.82 (0.8) & 12.22 (1.3) & 12.69 (1.6) &   9.94 (1) & 10.22 (1.1) &   9.68 (1) &  10.48 (1) & 8.06 (0.8) &   9.78 (1) \\

           &       Rank &  7.66, 2.8 &    3.28, 0 &     3.10, 0 &    5.53, 0 &    4.88, 0 &    6.71, 0 & 10.89, 25.6 &  7.57, 3.6 &    6.01, 0 \\

       0.1 &        Est & 3.53 (0.4) & 5.62 (0.7) & 5.72 (0.8) & 4.52 (0.5) & 4.61 (0.6) & 4.29 (0.5) & 4.31 (0.5) & 3.61 (0.4) &  4.40 (0.5) \\

           &       Pred &  12.07 (1) & 16.07 (1.4) & 16.56 (1.6) & 13.54 (1.2) & 13.78 (1.2) & 13.15 (1.1) & 14.36 (1.1) &  12.48 (1) & 13.37 (1.1) \\

           &       Rank & 11.18, 18.6 &    6.21, 0 &  5.95, 0.2 &  8.12, 4.6 &  7.64, 2.2 & 9.23, 33.2 &   14.06, 0 &  7.85, 5.2 &  8.35, 8.8 \\

       0.2 &        Est & 5.59 (0.6) & 6.84 (0.7) & 7.01 (0.8) & 5.88 (0.6) & 5.97 (0.6) & 5.67 (0.6) & 6.02 (0.6) & 5.54 (0.6) & 5.81 (0.6) \\

           &       Pred & 15.5 (1.1) & 17.17 (1.3) & 17.58 (1.4) & 15.53 (1.1) & 15.7 (1.2) & 15.27 (1.1) & 17.23 (1.3) & 15.36 (1.1) & 15.41 (1.1) \\

           &       Rank & 10.87, 20.8 &    9.00, 26.8 & 8.78, 19.8 & 10.03, 58.8 & 9.71, 55.2 &  10.95, 26 &   16.25, 0 & 9.38, 47.2 &  10.13, 58 \\

       0.3 &        Est &  6.40 (0.7) & 6.83 (0.8) & 6.94 (0.8) & 6.22 (0.7) & 6.29 (0.7) & 6.11 (0.7) & 6.81 (0.7) & 6.05 (0.7) & 6.19 (0.7) \\

           &       Pred & 16.40 (1.2) & 16.86 (1.2) & 17.09 (1.4) & 15.96 (1.2) & 16.06 (1.2) & 15.88 (1.2) & 18.51 (1.2) & 15.91 (1.2) & 15.89 (1.2) \\

           &       Rank & 10.69, 30.2 & 9.81, 81.4 & 9.67, 68.4 & 10.33, 64.2 & 10.14, 73.6 & 11.2, 10.8 &      17.00, 0 &   9.86, 86 &  10.39, 59 \\
\hline
           &            &                                                                                   \multicolumn{ 9}{|c}{$\rho=0.5$} \\
\hline
      0.05 &        Est &  0.80 (0.1) & 1.19 (0.1) & 1.24 (0.1) & 0.92 (0.1) & 0.94 (0.1) & 0.87 (0.1) & 0.85 (0.1) & 0.84 (0.1) & 0.86 (0.1) \\

           &       Pred &  12.36 (1) & 16.84 (1.5) &  17.77 (2) & 13.56 (1.2) & 13.80 (1.3) & 12.97 (1.1) & 13.25 (1.1) & 12.96 (1.1) & 12.99 (1.1) \\

           &       Rank & 12.83, 0.8 &    6.02, 0 &    5.63, 0 &  8.01, 5.2 &   7.60, 1.8 & 9.18, 30.4 &   13.28, 0 &  7.88, 6.4 & 9.05, 26.8 \\

       0.1 &        Est & 1.19 (0.1) & 1.35 (0.1) & 1.38 (0.2) & 1.17 (0.1) & 1.19 (0.1) & 1.13 (0.1) &  1.20 (0.1) & 1.15 (0.1) & 1.13 (0.1) \\

           &       Pred & 16.21 (1.2) & 17.42 (1.2) & 17.82 (1.5) & 15.72 (1.1) & 15.86 (1.2) & 15.37 (1.1) & 16.82 (1.1) & 15.74 (1.1) & 15.38 (1.1) \\

           &       Rank &   15.21, 0 &   9.15, 32 & 8.93, 24.6 & 10.02, 57.8 &  9.80, 56.4 & 10.98, 23.8 &   15.83, 0 & 9.39, 47.2 &  10.37, 52 \\

       0.2 &        Est & 1.36 (0.1) & 1.29 (0.1) & 1.29 (0.1) & 1.23 (0.1) & 1.24 (0.1) & 1.22 (0.1) & 1.41 (0.1) & 1.23 (0.1) & 1.22 (0.1) \\

           &       Pred & 17.63 (1.4) & 16.59 (1.2) & 16.67 (1.2) & 16.08 (1.1) & 16.16 (1.1) & 16.1 (1.1) &   19.00 (1.3) & 16.15 (1.1) & 15.98 (1.1) \\

           &       Rank & 12.59, 0.4 & 9.98, 98.2 &   9.95, 94 & 10.28, 72.2 &  10.15, 85 & 11.16, 10.2 &   17.34, 0 & 9.99, 98.8 &   10.40, 62 \\

       0.3 &        Est & 1.32 (0.1) & 1.25 (0.1) & 1.25 (0.1) & 1.22 (0.1) & 1.23 (0.1) & 1.22 (0.1) & 1.47 (0.1) & 1.22 (0.1) & 1.21 (0.1) \\

           &       Pred & 17.08 (1.3) & 16.22 (1.2) & 16.22 (1.2) & 15.95 (1.2) & 16.03 (1.2) &   16.00 (1.2) & 19.47 (1.3) &   16.00 (1.2) & 15.89 (1.2) \\

           &       Rank &  10.92, 14 &   10.00, 99.8 &   10.00, 99.4 & 10.21, 80.4 & 10.12, 89.8 &  11.04, 17 &   17.83, 0 &   10.00, 99.8 & 10.28, 73.8 \\
\hline
           &            &                                                                                   \multicolumn{ 9}{|c}{$\rho=0.1$} \\
\hline
      0.05 &        Est & 0.64 (0.1) & 0.86 (0.1) &  0.90 (0.1) & 0.69 (0.1) &  0.70 (0.1) & 0.65 (0.1) & 0.65 (0.1) & 0.67 (0.1) & 0.65 (0.1) \\

           &       Pred &  13.25 (1) & 17.39 (1.4) & 18.2 (1.8) & 14.12 (1.2) & 14.41 (1.3) & 13.50 (1.1) & 13.71 (1.1) & 14.01 (1.2) & 13.47 (1.1) \\

           &       Rank &   14.28, 0 &  6.64, 0.6 &    6.21, 0 &   8.48, 12 &  7.95, 5.6 & 9.64, 39.8 &   13.62, 0 &    8.10, 10 & 9.64, 35.4 \\

       0.1 &        Est & 0.85 (0.1) &  0.90 (0.1) & 0.91 (0.1) & 0.81 (0.1) & 0.82 (0.1) & 0.79 (0.1) & 0.85 (0.1) & 0.81 (0.1) & 0.78 (0.1) \\

           &       Pred & 16.82 (1.2) & 17.32 (1.3) & 17.65 (1.4) & 15.88 (1.2) & 16.01 (1.2) & 15.59 (1.2) & 17.16 (1.2) & 16.05 (1.2) & 15.54 (1.2) \\

           &       Rank &   16.51, 0 &   9.46, 52 &   9.27, 40 & 10.16, 63.6 & 9.92, 64.8 & 11.13, 15.4 &   16.09, 0 & 9.61, 63.2 & 10.55, 43.6 \\

       0.2 &        Est & 0.97 (0.1) & 0.85 (0.1) & 0.85 (0.1) & 0.82 (0.1) & 0.83 (0.1) & 0.82 (0.1) & 0.96 (0.1) & 0.82 (0.1) & 0.81 (0.1) \\

           &       Pred & 18.57 (1.4) & 16.39 (1.2) & 16.44 (1.2) & 15.99 (1.1) & 16.07 (1.1) &   16.00 (1.1) & 18.99 (1.3) & 16.02 (1.2) & 15.85 (1.1) \\

           &       Rank &   15.36, 0 &    10.00, 100 &   9.98, 97 & 10.22, 78.4 & 10.18, 84.8 &  11.09, 12 &   17.38, 0 &    10.00, 100 & 10.34, 67.4 \\

       0.3 &        Est & 0.92 (0.1) & 0.83 (0.1) & 0.83 (0.1) & 0.82 (0.1) & 0.83 (0.1) & 0.82 (0.1) &    1.00 (0.1) & 0.82 (0.1) & 0.82 (0.1) \\

           &       Pred & 17.7 (1.5) & 16.15 (1.2) & 16.17 (1.2) & 15.97 (1.2) & 16.05 (1.2) & 16.03 (1.2) & 19.64 (1.3) & 15.99 (1.2) & 15.90 (1.2) \\

           &       Rank &   11.25, 3 &    10.00, 100 & 10.01, 99.4 & 10.18, 82.8 & 10.14, 87.6 & 10.95, 17.2 &   17.86, 0 &    10.00, 100 &  10.26, 76 \\
\hline
\multicolumn{ 2}{c}{Time} &      17.52 &       0.02 &       0.02 &       0.15 &       0.18 &       0.15 &       0.15 &       3.93 &       2.06 \\
\hline
\end{tabular}
\end{tiny}
\end{center}
\end{table}

\begin{table}
\centering
\caption{\label{table2} Comparison of estimation, prediction and rank determination performances of various reduced-rank estimators using Model II ($n=20,p=100,q=25,r^*=5,r_x=10$). The layout of the table is the same as in Table \ref{table1}.}
\begin{tiny}
\begin{tabular}{ll|rrrrrrrrr}
\hline
       $b$ &        ERR &                                                                                       \multicolumn{ 9}{|c}{Method} \\

           &            & $\M{NNP}^O$ & $\M{RSC}^O$ & $\M{RSC}^C$ & $\M{ANN}^{O}_2$ & $\M{ANN}^{C}_2$ & $\M{ANN}^{O}_1$ & $\M{ANN}^{O}_0$ & $\M{RoRR}^O$ & $\M{RoANN}^{O}_2$ \\
\hline
           &            &                                                                                   \multicolumn{ 9}{|c}{$\rho=0.9$} \\
\hline
      0.05 &        Est & 1.14 (0.2) & 1.15 (0.2) & 1.15 (0.2) & 1.15 (0.2) & 1.15 (0.2) & 1.15 (0.2) & 1.15 (0.2) & 1.15 (0.2) & 1.15 (0.2) \\

           &       Pred & 34.12 (4.1) &  31.49 (4) & 32.32 (4.6) & 29.08 (3.5) & 29.65 (3.5) & 28.95 (3.4) & 35.28 (4.2) & 29.75 (3.5) & 28.89 (3.5) \\

           &       Rank &  7.38, 1.6 & 4.73, 73.4 & 4.63, 64.4 & 5.19, 71.2 & 4.99, 73.2 &  5.70, 36.4 &    8.26, 0 & 4.83, 83.2 &  5.30, 65.8 \\

       0.1 &        Est &  4.50 (0.7) & 4.52 (0.7) & 4.52 (0.7) & 4.52 (0.7) & 4.52 (0.7) & 4.52 (0.7) & 4.52 (0.7) & 4.52 (0.7) & 4.52 (0.7) \\

           &       Pred & 37.03 (4.4) & 30.83 (3.6) &   31.00 (3.8) & 30.01 (3.6) & 30.40 (3.8) & 30.15 (3.7) & 38.86 (4.8) & 30.25 (3.6) & 29.88 (3.6) \\

           &       Rank &  7.57, 2.8 & 4.99, 99.2 & 4.97, 96.6 & 5.22, 79.4 & 5.13, 87.4 & 5.63, 42.8 &    8.77, 0 & 4.99, 99.4 & 5.27, 74.6 \\

       0.2 &        Est & 18.06 (2.7) & 17.92 (2.6) & 17.92 (2.6) & 17.92 (2.6) & 17.92 (2.6) & 17.92 (2.6) & 17.93 (2.6) & 17.92 (2.6) & 17.92 (2.6) \\

           &       Pred & 34.90 (5.5) & 29.79 (3.6) & 29.81 (3.6) & 29.55 (3.6) & 29.81 (3.6) & 29.67 (3.6) & 40.08 (4.8) & 29.59 (3.6) & 29.52 (3.6) \\

           &       Rank & 6.33, 39.4 &     5.00, 100 &    5.00, 99.8 & 5.22, 79.2 &   5.12, 89 & 5.53, 51.4 &    9.02, 0 &     5.00, 100 & 5.29, 72.8 \\

       0.3 &        Est & 40.42 (5.6) & 40.58 (5.9) & 40.58 (5.9) & 40.58 (5.9) & 40.58 (5.9) & 40.58 (5.9) & 40.58 (5.9) & 40.58 (5.9) & 40.58 (5.9) \\

           &       Pred & 32.29 (5.5) & 30.18 (3.5) & 30.22 (3.6) & 30.06 (3.5) & 30.32 (3.6) & 30.16 (3.5) &  41.29 (5) & 30.04 (3.5) & 30.03 (3.5) \\

           &       Rank &   5.48, 75 &     5.00, 100 &    5.00, 99.6 &   5.18, 82 & 5.15, 87.4 & 5.46, 56.2 &    9.07, 0 &     5.00, 100 & 5.26, 75.2 \\
\hline
           &            &                                                                                   \multicolumn{ 9}{|c}{$\rho=0.5$} \\
\hline
      0.05 &        Est & 1.12 (0.2) & 1.12 (0.2) & 1.12 (0.2) & 1.12 (0.2) & 1.12 (0.2) & 1.12 (0.2) & 1.12 (0.2) & 1.12 (0.2) & 1.12 (0.2) \\

           &       Pred & 34.98 (4.2) & 31.01 (3.9) & 31.57 (4.2) & 29.14 (3.6) & 29.65 (3.8) & 29.12 (3.6) & 35.92 (4.7) & 29.73 (3.6) &   29.00 (3.6) \\

           &       Rank &    7.97, 0 &  4.90, 89.8 &   4.79, 78 &    5.20, 77 & 5.06, 79.2 & 5.68, 37.6 &    8.43, 0 & 4.93, 93.2 & 5.32, 67.8 \\

       0.1 &        Est &  4.50 (0.6) & 4.53 (0.6) & 4.53 (0.6) & 4.53 (0.6) & 4.53 (0.6) & 4.53 (0.6) & 4.53 (0.6) & 4.53 (0.6) & 4.53 (0.6) \\

           &       Pred & 37.83 (4.6) & 30.16 (3.5) & 30.29 (3.6) & 29.58 (3.3) & 29.93 (3.6) & 29.73 (3.4) & 38.98 (4.2) & 29.76 (3.4) & 29.52 (3.4) \\

           &       Rank &  7.76, 7.2 & 4.99, 99.4 &   4.99, 98 &   5.19, 81 &    5.10, 90 & 5.63, 42.2 &    8.86, 0 & 4.99, 99.4 & 5.27, 73.6 \\

       0.2 &        Est & 17.90 (2.6) & 18.03 (2.7) & 18.03 (2.7) & 18.03 (2.7) & 18.03 (2.7) & 18.03 (2.7) & 18.04 (2.7) & 18.03 (2.7) & 18.03 (2.7) \\

           &       Pred & 34.84 (6.1) & 30.05 (3.6) & 30.07 (3.6) & 29.89 (3.6) & 30.15 (3.6) & 29.99 (3.6) &  40.75 (5) & 29.88 (3.6) & 29.85 (3.6) \\

           &       Rank & 6.25, 51.4 &     5.00, 100 &    5.00, 99.8 & 5.21, 79.8 & 5.14, 87.8 & 5.51, 53.4 &    9.07, 0 &     5.00, 100 &   5.27, 75 \\

       0.3 &        Est & 40.44 (6.1) &  40.13 (6) &  40.13 (6) &  40.13 (6) &  40.13 (6) &  40.13 (6) &  40.13 (6) &  40.13 (6) &  40.13 (6) \\

           &       Pred & 31.55 (4.9) & 30.32 (3.1) & 30.32 (3.1) & 30.23 (3.1) & 30.44 (3.1) & 30.3 (3.1) &   41.30 (4) & 30.18 (3.1) & 30.2 (3.1) \\

           &       Rank & 5.34, 85.6 &     5.00, 100 &     5.00, 100 & 5.23, 77.8 &   5.11, 89 & 5.46, 56.8 &    9.16, 0 &     5.00, 100 & 5.32, 70.2 \\
\hline
           &            &                                                                                   \multicolumn{ 9}{|c}{$\rho=0.1$} \\
\hline
      0.05 &        Est & 1.13 (0.2) & 1.13 (0.2) & 1.13 (0.2) & 1.13 (0.2) & 1.13 (0.2) & 1.13 (0.2) & 1.13 (0.2) & 1.13 (0.2) & 1.13 (0.2) \\

           &       Pred & 35.17 (4.2) & 31.02 (3.8) & 31.86 (4.5) & 29.25 (3.5) & 29.7 (3.6) & 29.24 (3.5) & 36.04 (4.3) & 29.85 (3.5) & 29.10 (3.6) \\

           &       Rank &    8.07, 0 & 4.88, 88.2 & 4.77, 76.8 & 5.19, 75.6 &  5.10, 78.8 &    5.70, 35 &    8.41, 0 &   4.91, 91 & 5.32, 65.4 \\

       0.1 &        Est & 4.55 (0.7) & 4.47 (0.7) & 4.47 (0.7) & 4.47 (0.7) & 4.47 (0.7) & 4.47 (0.7) & 4.47 (0.7) & 4.47 (0.7) & 4.47 (0.7) \\

           &       Pred & 37.73 (4.6) & 30.24 (3.5) & 30.33 (3.7) & 29.72 (3.5) & 30.04 (3.5) & 29.81 (3.5) & 38.88 (4.4) & 29.8 (3.5) & 29.64 (3.4) \\

           &       Rank &  7.89, 5.4 &    5.00, 99.8 & 4.99, 99.2 & 5.22, 79.4 & 5.12, 89.2 & 5.62, 43.4 &    8.84, 0 &     5.00, 100 & 5.29, 72.6 \\

       0.2 &        Est & 18.1 (2.6) & 17.92 (2.6) & 17.92 (2.6) & 17.92 (2.6) & 17.92 (2.6) & 17.92 (2.6) & 17.92 (2.6) & 17.92 (2.6) & 17.92 (2.6) \\

           &       Pred & 34.48 (5.9) & 30.13 (3.4) & 30.13 (3.4) & 29.98 (3.4) & 30.24 (3.4) & 30.09 (3.4) & 40.76 (4.7) & 29.97 (3.4) & 29.93 (3.4) \\

           &       Rank & 6.18, 54.4 &     5.00, 100 &     5.00, 100 &   5.23, 78 & 5.13, 88.2 & 5.51, 52.6 &    9.07, 0 &     5.00, 100 &  5.30, 71.6 \\

       0.3 &        Est & 40.48 (6.2) & 40.48 (5.6) & 40.48 (5.6) & 40.48 (5.6) & 40.48 (5.6) & 40.48 (5.6) & 40.48 (5.6) & 40.48 (5.6) & 40.48 (5.6) \\

           &       Pred & 31.31 (5.1) & 30.06 (3.4) & 30.11 (3.5) & 29.96 (3.4) & 30.28 (3.5) & 30.05 (3.4) & 41.17 (4.6) & 29.91 (3.4) & 29.94 (3.4) \\

           &       Rank & 5.31, 86.2 &     5.00, 100 &    5.00, 99.6 &   5.24, 76 &   5.15, 87 & 5.46, 55.2 &    9.08, 0 &     5.00, 100 & 5.34, 67.8 \\
\hline
\multicolumn{ 2}{c}{Time} &      20.22 &       0.02 &       0.08 &       0.15 &       0.26 &       0.15 &       0.15 &       3.95 &       2.49 \\
\hline
\end{tabular}
\end{tiny}
\end{table}

Tables \ref{table1} and \ref{table2} report the simulation results for Model I and II, respectively. We summarize our findings as follows.

\begin{itemize}
  \item We first examine the performance of $\M{ANN}_\gamma$ with different $\gamma$ values. The $\M{ANN}_0$, which does not use the adaptive weights, is not as accurate as the ANN with adaptive weights in both rank estimation and prediction in general. Its behavior is similar to that of the NNP: they both tend to overestimate the rank but may perform well when the signal is weak and the model dimension is not high (Model I). The performance of both $\M{ANN}_1$ and $\M{ANN}_2$ is substantially better than $\M{ANN}_0$ with the aid of adaptive weights. We have also experimented with other $\gamma$ values, and as expected, the ANN estimator behave more and more similar to the RSC as $\gamma$ increases. Our results show that $\gamma=2$ is generally a very good choice. Henceforth we always refer to $\M{ANN}_2$ in the following comparisons with other reduced-rank methods.
  \item The $\M{ANN}_2$ generally outperforms the RSC in both estimation and prediction. The improvement can be substantial, especially when the signal is weak or moderate and the correlation among the predictors is high. For rank determination, both the ANN and the RSC have similar excellent behavior when the signal is moderate to strong and the correlation among the predictors is weak to moderate. We notice that the ANN estimator tends to slightly overestimate the rank; however, the overestimation is generally negligible. When the signal strength is moderate to large and correlation among predictors is weak to moderate ANN does slightly worse in terms of rank selection than RSC yet is able to maintain a small gain in estimation and prediction. This gain is due to the shrinkage effect using soft-thresholding. While RSC is based on hard-thresolding and keeps the first few leading SVD layers completely as signal, in reality all the SVD layers are contaminated by noise, and hence the shrinkage estimation is effective. When the signal is weak and the correction among the predictors is high, the rank determination performance of the RSC can be much worse than that of the ANN, which is also reflected in their differences in prediction and estimation.

  \item The $\M{ANN}_2$ also outperforms the NNP in general, and is more parsimonious than the NNP in both rank reduction and computation. Only when the signal is weak and the correlation among predictors is very high, may the NNP method slightly outperform the ANN (and hence RSC) in estimation and prediction. However, this gain exacts a high cost that may make it not worthwhile, as the NNP often excessively overestimates the rank and is much harder to compute. Our findings regarding the NNP agrees with those reported in \citet{bunea2011}. Note that \citet{bunea2011} also proposed a calibrated NNP method for rank determination, which showed very similar behavior as that of the RSC. (Hence it is not reported here.)

  \item Adding an $l_2$ penalty usually boosts the performance of the reduced-rank methods. The $\M{RoRR}$ may substantially outperform its non-robustified counterpart RSC, especially when the correlation is high, showing the power of shrinkage estimation. The $\M{RoANN}_2$ only slightly outperforms $\M{ANN}_2$, because $\M{ANN}_2$ itself performs adaptive shrinkage estimation and thus the gain from the overall shrinkage induced by the $l_2$ penalty is much limited. $\M{RoRR}$, $\M{ANN}_2$ and $\M{RoANN}_2$ have comparable performance in most cases and are the best methods.
  \item NNP is much more computationally expensive than the other reduced-rank methods. Both RSC and ANN are very fast to compute, and in practice they can always be computed together as they rely on the same SVD operation. Adding an $l_2$ penalty increases computation time. $\M{RoRR}$ takes longer computation time than $\M{RoANN}$, for the reason explained in Section~\ref{sec:robust}.
\end{itemize}
Overall, the ANN approach is preferable to both the RSC and the NNP methods, especially for cases when the data are noisy and the correlation among the predictors is high. Adding extra $l_2$ penalty is certainly worthwhile, especially for RSC, but it may incur a lot of computation efforts.

\subsection{An application in Genomics}

We consider a breast cancer data set \citep{witten2009,bunea2010}, which consists of the gene expression measurements and the comparative genomic hybridization (CGH) measurements for $n = 89$ subjects. The data set is available in the R package \emph{PMA}, and a detailed description can be found in \citet{chin2006}.

Prior studies have demonstrated a link between DNA copy-number changes and cancer risk \citep{pollack2002,peng2010}. It is thus of interest to examine the relationship between DNA copy number variations (CNVs) and gene expression profiles (GEPs), for which multivariate regression methods can be useful. Biologically, it makes sense to regress GEPs on CNVs since the latter play an important role in regulating the former. The reverse approach of regressing CNVs on GEPs is also meaningful, in that the resulting predictive model may identify functionally relevant CNAs; this approach has been shown to be promising in enhancing the limited CGH analysis with the wealth of GEP data \citep{Geng2011,Zhou2012}. We thus try both approaches, i.e., setting 1: designate the CNVs on the CGH spots of a chromosome as predictors ($\B{X}_{n\times p}$), and the GEPs of the same chromosome as responses ($\B{Y}_{n\times q}$), and setting 2: reverse the roles of $\B{X}$ and $\B{Y}$. Both the responses and predictors are centered and standardized.

We focus on chromosomes 14 and 21, which were previously studied in \citet{bunea2010}. In this study, $p$ and $q$ are either comparable or much larger than $n=89$. To alleviate the high-dimensionality problem, the reduced-rank methods are appealing for identifying a few linear combinations of predictors for optimally predicting the response variables. We perform model estimation using various reduced-rank methods, with the tuning parameters selected by ten-fold cross validation. (We obtained similar results with five-fold and fifteen-fold cross validations.) The cross validation error rate (CVE) (averaged over the number of response variables and the sample size) is then used to compare the predictive performance of different penalization schemes. Note that $\mbox{CVE}=1$ corresponds to a null model with zero coefficient matrix. Table \ref{table:gene} reports the CVE, the estimated rank and the computation time.

\begin{table}
\centering
\caption{\label{table:gene} Performances of reduced-rank estimators on breast cancer data. Setting 1 regresses GEPs on CNVs, and setting 2 regresses CNVs on GEPs.}
{
\begin{scriptsize}
\begin{tabular}{l|rrrrr}
\hline
           &                                   \multicolumn{ 5}{|c}{Method} \\

           & $\M{NNP}^C$ & $\M{RSC}^C$ & $\M{RoR}^C$ & $\M{ANN}^{C}_2$ & $\M{RoANN}^{C}_2$ \\
\hline
           & \multicolumn{ 5}{|c}{Setting 1: Chromosome 14 ($q=641$, $p=76$)} \\
\hline
       CVE &       0.84 &       1.00 &       0.85 &       0.96 &       0.92 \\

      Rank &          4 &          0 &          7 &          1 &          1 \\

      Time &     2190.4 &        5.5 &      689.0 &       24.6 &       38.2 \\
\hline
           & \multicolumn{ 5}{|c}{Setting 2: Chromosome 14 ($q=76$, $p=641$)} \\
\hline
       CVE &       0.69 &       0.59 &       0.59 &       0.58 &       0.58 \\

      Rank &         23 &          5 &          5 &         11 &         17 \\

      Time &     2158.5 &       13.5 &      186.7 &       13.8 &       16.8 \\
\hline
\hline
           & \multicolumn{ 5}{|c}{Setting 1: Chromosome 21 ($q=227$, $p=44$)} \\
\hline
       CVE &       0.84 &       0.95 &       0.83 &       0.87 &       0.86 \\

      Rank &          3 &          1 &          2 &          1 &          1 \\

      Time &     1266.4 &        0.5 &       66.2 &        2.6 &        8.4 \\
\hline
           & \multicolumn{ 5}{|c}{Setting 2: Chromosome 21 ($q=44$, $p=227$)} \\
\hline
       CVE &       0.69 &       0.65 &       0.63 &       0.62 &       0.62 \\

      Rank &          5 &          1 &          1 &          1 &          2 \\

      Time &     1226.0 &        0.5 &       28.6 &        0.7 &        2.9 \\
\hline
\end{tabular}
\end{scriptsize}
}
\end{table}

In setting 1, the SNR is very low as reflected by the CVEs being close to 1. The RSC may fail to pick up any signal (for chromosome 14), while the other methods, especially NNP and RoRR, perform better owing to the power of shrinkage estimation. In setting 2, the SNR is relatively higher, and the ANN and RoANN have better prediction performance than all the other methods. In Section 4, we have shown that the prediction error is of the order $(r_x+q)r^*$, where $r_x \leq \min(n,p)$ is the rank of the design matrix $\B{X}$. This may partly explain why in setting 1 the prediction is always poorer than in setting 2, because $q$ is much larger than $p$ when the CNVs (GEPs) serve as predictors (responses). For rank estimation, the NNP method always yields higher rank estimates than the other methods; the ANN estimated rank is higher than that of RSC for chormosome 14 in setting 2, otherwise their rank estimates are similar. Adding $l_2$ penalty improves the predictive performance of the reduced-rank methods, especially the improvement from RSC to RoRR can be substantial. The robustified methods may also yield higher rank estimates. As borne out by the simulation study, such behaviors result from the hybridization between reduced-rank methods and ridge regression. Both NNP and RoRR are computationally intensive for large datasets, while RSC, ANN and RoANN are much faster to compute. Overall, it can be seen that the proposed ANN approach shows better performance than the RSC, with not much extra computational cost.


\section{Discussion}\label{sec:discussion}

There are several potential directions for future research. We have mainly considered the adaptive nuclear-norm penalization on $\B{XC}$ which yields a computationally efficient SVD-thresholding estimator for dimension reduction and shrinkage estimation. It is interesting to consider an ANN criterion that puts the adaptive nuclear-norm penalization directly on the coefficient matrix $\B{C}$, i.e., $\{\|\B{Y}-\B{XC}\|_F^2+\lambda\sum w_i d_i(\B{C})\}$. This criterion may be optimized by similar iterative SVD-thresholding method used in solving the NNP problem. Although the computation will be more intensive, the advantage is that this criterion would result in simultaneous adaptive rank reduction and shrinkage estimation. 

The proposed ANN method can serve as the building block to study a family of singular value penalties. This is based on the connection between an adaptive $l_1$-type penalty and many concave penalty functions such as SCAD \citep{fan2001} and bridge penalty \citep{knight2000}. Consider the general regression problem (\ref{criterion1}) with a general singular value penalty
$\mathcal{P}_\lambda(\B{C})=\sum_{i=1}^{h} p_{\lambda}(d_i)$, where $p_{\lambda}(\cdot)$ is a penalty function, e.g., $l_q$ bridge penalty $p_\lambda (|d_i|)=\lambda|d_i|^q$ ($0<q<1$) \citep{Rohde2011}. In this setup the optimization of (\ref{criterion1}) can be challenging. A promising approach is to adopt a local linear approximation \citep{zou2008},
$p_\lambda(|d_i|)\approx p_\lambda(|d_i^{(0)}|) + p_\lambda'(|d_i^{(0)}|)(|d_i|-|d_i^{(0)}|)$, for $d_i\approx d_i^{(0)}$, where $d_i^{(0)}$ is some initial estimator of $d_i$ which for example can be obtained by the LS method. It can be seen that for fixed $d_i^{(0)}$, up to a constant, the first-order approximated penalty admits exactly an ANN form. This suggests the ANN estimator with the weights $p_\lambda'(|d_i^{(0)}|)$ can be viewed as an one-step estimator of these problems, and these problems may be solved by an iteratively reweighted ANN approach. 


It is shown that incorporating an extra ridge penalty can induce further shrinkage and hence improve the reduced-rank estimation. When combined with the ANN penalty, such a criterion bears resemblance to the elastic-net criterion (Enet) \citep{zou2005} in univariate regression. It would be interesting to investigate the properties of the SVD-Enet approaches. Another pressing problem concerns further extending the regularized reduced-rank regression methods to generalized linear models and nonparametric regression models \citep{yee2003,li2007}. On the optimization aspect, it is interesting to study the usage of ANN in some classical sparse optimization areas, such as matrix completion \citep{candes2011}.

\appendix
\section*{Technical details}

\subsection*{Proof of Theorem \ref{thm:convexity}}

\begin{proof}
First we show by a counter example that if we have an index $k$ such that $w_{k} < w_{k+1}$, then $f(\cdot)$ is
non-convex. Let $\B{C}$ and $\B{D}$ be diagonal $p\times q$ matrices such that $c_{ii} = i$, for $i=1,...,h$, while $\B{D}$ equals
$\B{C}$ but with entries switched at positions $h-k+1$ and $h-k$ on the diagonal. It is then easy to verify that
\begin{align*}
f(\B{C}) = f(\B{D}) =& \sum_{i=1}^h w_i (h-i+1),\\
f(\frac{\B{C}+\B{D}}{2}) -  \frac{f(\B{C})}{2} -  \frac{f(\B{D})}{2}=&
(h-k+0.5)(w_k+w_{k+1}) - (h-k+1)w_k - (h-k) w_{k+1}\\
 =& 0.5(w_k+w_{k+1}) - w_k > 0,
\end{align*}
where $f(\cdot)$ is defined in (\ref{def:ann}). Therefore $f(\cdot)$ is
non-convex.

Next we prove that for $w_1 \geq \cdots \geq w_{h} \geq 0$, $f(\cdot)=\|\cdot\|_{w*}$ is a convex function. First consider the case that $w_h > 0$, and define the following function on $\Re^h$:
\begin{equation}\label{def:w}
w(\B{x}) = \sum_{i=1}^h w_i \left| \B{x} \right|_{\delta(i)},
\end{equation}
where $\delta$ is a permutation of $\{1, ..., h\}$ determined by $\B{x}$ such that
$\left| \B{x} \right|_{\delta(1)} \geq \left| \B{x} \right|_{\delta(2)} \geq \cdots \geq \left| \B{x} \right|_{\delta(h)}$, where $|\B{x}|$ is the vector of absolute values of $\B{x}$. We claim that $w(\cdot)$ is a \textit{symmetric gauge function}
(see \citet[Definition 7.4.23]{MatAnalysis} for reference),
i.e., it satisfies the following six conditions: (a) $w(\B{x}) \geq 0, \forall \B{x}$; (b) $w(\B{x}) = 0$ if and only if $\B{x} = 0$;
(c) $w(\alpha \B{x}) = |\alpha| w(\B{x})$, $\forall \alpha \in \Re$; (d) $w(\B{x}+\B{y}) \leq w(\B{x}) + w(\B{y})$; (e) $w(\B{x}) = w(|\B{x}|)$; (f) $w(\B{x}) = w(\tau(\B{x}))$ for any $\tau$ is a permutation of indices $\{1, ..., h\}$.

All conditions except (d) are trivial to verify. Now we prove (d). Let
$\delta$,$\sigma$,$\tau$ be permutations such that $\left\{\left| \B{x}+\B{y} \right|_{\delta(i)}\right\}$,
$\left\{\left| \B{x} \right|_{\sigma(i)}\right\}$ and
$\left\{\left| \B{y} \right|_{\tau(i)}\right\}$ are placed in non-increasing order respectively.
\begin{eqnarray*}
w(\B{x}+\B{y}) &=& \sum_{i=1}^h w_i \left| \B{x}+\B{y}\right| _{\delta(i)} \leq
\sum_{i=1}^h \left\{w_i \left| \B{x} \right|_{\delta(i)} +w_i \left| \B{y}\right| _{\delta(i)}\right\} \\
& \leq & \sum_{i=1}^h \left\{w_i \left| \B{x} \right|_{\sigma(i)} +w_i \left| \B{y}\right| _{\tau(i)}\right\}
 = w(\B{x}) + w(\B{y}).
\end{eqnarray*}
where the second inequality is due to the Hardy-Littlewood-P\'{o}lya inequality \citep{HardyLittlewoodPolya_Ineq}.


Then by a straightforward application of  \citet[Theorem 7.4.24]{MatAnalysis}, since $\|\B{C}\|_{w*} = w([d_1(\B{C}), d_2(\B{C}), ..., d_h(\B{C})]^\T)$,
$\|\cdot\|_{w*}$ defines a matrix norm and hence is a convex function.

For the case that $w_h = 0$, let $s$ to be the largest index such that $w_s > 0$. For $0 < \epsilon < w_s$, consider the perturbated $\tilde{w}$ that $\tilde{w}_i = w_i$, for $i=1,...,s$,
and $\tilde{w}_i = \epsilon$, for $i = s+1,...,h$. Then for any $\B{A}, \B{B} \in \Re^{n\times q}$, $\|\frac{\B{A}+\B{B}}{2}\|_{\tilde{w}*} \leq \frac{\|\B{A}\|_{\tilde{w}*}}{2}
+ \frac{\|\B{B}\|_{\tilde{w}*}}{2}$. By taking $\epsilon \to 0$, $\|\frac{\B{A}+\B{B}}{2}\|_{w*} \leq \frac{\|\B{A}\|_{w*}}{2} + \frac{\|\B{B}\|_{w*}}{2}$.
Therefore $\left\|\cdot\right\|_{w*}$ is convex.
\end{proof}

\subsection*{Proof of Theorem \ref{thm:optimality}}

\begin{proof}
We first prove that $\hat{\B{C}}$ is indeed a global optimal solution to (\ref{eq:DC}).
Since the penalty term only depends on the singular values of $\B{C}$, by letting
$\B{g}=\{g_i\}_{i=1}^h = \B{d}(\B{C})$ (which implies the entries of $\B{g}$ are in non-increasing order),
(\ref{eq:DC}) can be equivalently written as:
$$
\min_{\B{g}: g_1\geq \cdots \geq g_h \geq 0}
\left\{ \min_{\tiny \begin{array}{c}\B{C}\in \Re^{n \times q}\\ \B{d}(\B{C}) = \B{g}\end{array}}
\left\{ \frac{1}{2} \| \B{Y} - \B{C} \|_F^2 \right\} + \lambda \sum_{i=1}^h w_i g_i \right\}.
$$
For the inner minimization, we have the inequality
\begin{eqnarray*}
\|\B{Y}- \B{C}\|_F^2 &=& \tr(\B{Y} - \B{C})(\B{Y} - \B{C})^\T \\
&=& \tr(\B{Y} \B{Y}^\T) - 2\tr(\B{Y} \B{C}^\T) + \tr(\B{C} \B{C}^\T) \\
&=& \sum_{i=1}^h d_i^2(\B{Y}) - 2\tr(\B{Y} \B{C}^\T) + \sum_{i=1}^h g_i^2 \\
&\geq& \sum_{i=1}^h d_i^2(\B{Y}) - 2\B{d}(\B{Y})^\T \B{g}+ \sum_{i=1}^h g_i^2.
\end{eqnarray*}
The last inequality is due to von Neumann's trace inequality. See \citet{Mirsky1976} for a proof. The equality holds when $\B{C}$ admits the singular value decomposition $\B{C} = \B{U} \Diag(\B{g}) \B{V}^\T$, where $\B{U}$ and $\B{V}$ are defined in (\ref{svdQ}) as the left and right orthonormal matrices in the SVD of $\B{Y}$. Then the optimization is reduced to
\begin{equation}\label{eq:strictconvex}
\min_{\B{g}: g_1\geq \cdots \geq g_h \geq 0}
\left\{ \sum_{i=1}^h \left( \frac{1}{2}g_i^2 - [d_i(\B{Y}) - \lambda w_i ] g_i + \frac{1}{2} d_i^2(\B{Y})\right) \right\}.
\end{equation}
The objective function is completely separable and takes minimum when $g_i = (d_i(\B{Y})-\lambda w_i)_+$. This is a feasible solution because $\{d_i(\B{Y})\}$ is in non-increasing order, while $\{w_i\}$ is in non-decreasing order. Therefore
$
\hat{\B{C}} = \mathcal{S}_{\lambda \B{w}} (\B{Y}) = \B{U} \Diag\{(\B{d}(\B{Y})-\lambda \B{w})_+\} \B{V}^\T
$
is a global optimal solution to (\ref{eq:DC}). The uniqueness follows by the equality condition for von Neumann's trace inequality
when $\B{Y}$ has a unique SVD, and the uniqueness of the strictly convex optimization (\ref{eq:strictconvex}). This concludes the proof.
\end{proof}

\subsection*{Proof of Theorem \ref{th:rank}}

\begin{proof}[ of Lemma \ref{lemma:rank1}]
By (\ref{est:rank}), $\hat{r} > s \Longleftrightarrow d_{s+1}(\B{PY}) > \lambda^{\frac{1}{\gamma+1}}$ and $\hat{r} < s \Longleftrightarrow d_{s}(\B{PY})\leq \lambda^{\frac{1}{\gamma+1}}$. Then
\begin{align*}
\M{P}(\hat{r}\neq s) = \M{P}\{d_{s+1}(\B{PY})>\lambda^{\frac{1}{\gamma+1}}\M{  or  } d_{s}(\B{PY})\leq \lambda^{\frac{1}{\gamma+1}}\}.
\end{align*}

Based on the Weyl's inequalities on singular values \citep{franklin2000} and observing that $\B{PY}=\B{XC}+\B{PE}$, we have $d_1(\B{PE})\geq d_{s+1}(\B{PY})-d_{s+1}(\B{XC})$ and $d_1(\B{PE})\geq d_{s}(\B{XC})-d_{s}(\B{PY})$. Hence $d_{s+1}(\B{PY}) > \lambda^{\frac{1}{\gamma+1}}$ implies $d_1(\B{PE})\geq \lambda^{\frac{1}{\gamma+1}}-d_{s+1}(\B{XC})$, and $d_{s}(\B{PY})\leq \lambda^{\frac{1}{\gamma+1}}$ implies $d_1(\B{PE})\geq d_s(\B{XC})-\lambda^{\frac{1}{\gamma+1}}$. It then follows that
\begin{align*}
\M{P}(\hat{r}\neq s) \leq \M{P}\{d_1(\B{PE})\geq \min(\lambda^{\frac{1}{\gamma+1}}-d_{s+1}(\B{XC}),d_s(\B{XC})-\lambda^{\frac{1}{\gamma+1}})\}.
\end{align*}
Finally, note that $\min(\lambda^{\frac{1}{\gamma+1}}-d_{s+1}(\B{XC}),d_s(\B{XC})-\lambda^{\frac{1}{\gamma+1}})\ge \delta\lambda^{\frac{1}{\gamma+1}}$. This completes the proof.
\end{proof}


\begin{lemma}[Lemma 3 of \citet{bunea2011}]\label{lemma:rank2}
Let $r_x=rank(\B{X})$ and suppose Assumption 2 holds. Then for any $t>0$, $E[d_1(\B{PE})]\leq \sigma(\sqrt{r_x}+\sqrt{q})$, and $\M{P}\{d_1(\B{PE})\geq E[d_1(\B{PE})]+\sigma t\}\leq \exp(-t^2/2)$.

\end{lemma}


\begin{proof}[ of Theorem \ref{th:rank}]
When $d_{r^*}(\B{XC})>2\lambda^{\frac{1}{\gamma+1}}$, we have
\begin{align*}
&d_{r^*}(\B{XC})> 2\lambda^{\frac{1}{\gamma+1}} \geq (1+\delta)\lambda^{\frac{1}{\gamma+1}},\M{ and }
d_{r^*+1}(\B{XC}) =0 \leq (1-\delta)\lambda^{\frac{1}{\gamma+1}},
\end{align*}
for some $0<\delta\leq 1$. It can be seen that the effective rank $s$ defined in Lemma \ref{lemma:rank1} equals to the true rank, i.e., $s=r^*$, and
$\min(\lambda^{\frac{1}{\gamma+1}}-d_{r^*+1}(\B{XC}),d_{r^*}(\B{XC})-\lambda^{\frac{1}{\gamma+1}})\ge \delta\lambda^{\frac{1}{\gamma+1}}$. It then follows by using the properties of Gaussian errors presented in Lemma \ref{lemma:rank2} that
\begin{align*}
\M{P}(\hat{r}=r^*)
&\geq 1- \M{P}(d_1(\B{PE})\geq \delta\lambda^{\frac{1}{\gamma+1}})\\
&=1-\M{P}\{d_1(\B{PE})\geq (1+\theta)\sigma(\sqrt{r_x}+\sqrt{q})\}\\
&\geq 1-\exp(-\theta^2(r_x+q)/2) \rightarrow 1.
\end{align*}
\end{proof}

\subsection*{Proof of Theorem \ref{bound:th2}}

\begin{proof}

By the definitions of $\hat{\B{C}}$ in (\ref{astest}),
$$
\|\B{Y}-\B{X}\hat{\B{C}}\|_F^2+2\lambda\sum w_id_i(\B{X}\hat{\B{C}})
\leq \|\B{Y}-\B{X}\B{B}\|_F^2+2\lambda\sum w_i  d_i(\B{XB}),
$$
for any $p\times q$ matrix $\B{B}$. Note that
\begin{align*}
&\|\B{Y}-\B{X}\hat{\B{C}}\|_F^2
=\|\B{Y}-\B{XC}\|_F^2
+\|\B{X}\hat{\B{C}}-\B{XC}\|_F^2
+2<\B{E},\B{XC}-\B{X}\hat{\B{C}}>_F,\\
&\|\B{Y}-\B{X}\B{B}\|_F^2
=\|\B{Y}-\B{XC}\|_F^2
+\|\B{X}\B{B}-\B{XC}\|_F^2
+2<\B{E},\B{XC}-\B{XB}>_F.
\end{align*}
Then we have
\begin{align}
&\|\B{X}\hat{\B{C}}-\B{XC}\|_F^2\notag\\
\leq& \|\B{X}\B{B}-\B{XC}\|_F^2
+2<\B{E},\B{X}\hat{\B{C}}-\B{XB}>_F
+2\lambda\{\sum w_i d_i(\B{XB})-\sum w_i d_i(\B{X}\hat{\B{C}})\}\notag\\
\leq& \|\B{X}\B{B}-\B{XC}\|_F^2
+2<\B{PE},\B{X}\hat{\B{C}}-\B{XB}>_F
+2\lambda\{\sum w_i d_i(\B{XB})-\sum w_i d_i(\B{X}\hat{\B{C}})\}\label{th2:eq1}\\
\leq& \|\B{X}\B{B}-\B{XC}\|_F^2
+2d_1(\B{PE})\|\B{X}\hat{\B{C}}-\B{XB}\|_*
+2\lambda\{\sum w_i d_i(\B{XB})-\sum w_i d_i(\B{X}\hat{\B{C}})\}\notag\\
\leq& \|\B{X}\B{B}-\B{XC}\|_F^2
+2d_1(\B{PE})\sqrt{r(\B{X}\hat{\B{C}}-\B{XB})}\|\B{X}\hat{\B{C}}-\B{XB}\|_F
+2\lambda\{\sum w_i d_i(\B{XB})-\sum w_i d_i(\B{X}\hat{\B{C}})\}.\notag
\end{align}

Now consider any $\B{B}$ with $r(\B{B})\leq \hat{r}$,
\begin{align*}
&\sum w_i d_i(\B{XB})-\sum w_i d_i(\B{X}\hat{\B{C}})\\
=&w_{\hat{r}}\sum_{i=1}^{\hat{r}}d_i(\B{XB})-w_{\hat{r}}\sum_{i=1}^{\hat{r}}d_i(\B{X}\hat{\B{C}})
+\sum_{i=1}^{\hat{r}} (w_{\hat{r}}-w_i)d_i(\B{X}\hat{\B{C}})-\sum_{i=1}^{\hat{r}} (w_{\hat{r}}-w_i)d_i(\B{XB}).
\end{align*}
By the definition of the adaptive weights in (\ref{eq:weights}), i.e., $w_i=d_i^{-\gamma}(\B{PY})$, we know that $w_{\hat{r}}-w_1\geq\cdots\geq w_{\hat{r}}-w_{\hat{r}-1}\geq0$. Therefore, both $p_1(\cdot)=\sum_{i=1}^{\hat{r}}d_i(\cdot)$ and $p_2(\cdot)=\sum_{i=1}^{\hat{r}}(w_{\hat{r}}-w_i)d_i(\cdot)$ satisfy the triangular inequality; see the proof of Theorem \ref{thm:convexity}. Moreover, based on Weyl's inequalities \citep{franklin2000} and $\B{PY}=\B{XC}+\B{PE}$, $d_{\hat{r}}(\B{PY})\geq d_{\hat{r}}(\B{XC})-d_1(\B{PE})$ and $d_{1}(\B{PY})\leq d_{1}(\B{XC})+d_1(\B{PE})$. It follows that
\begin{align*}
&\sum w_i d_i(\B{XB})-\sum w_i d_i(\B{X}\hat{\B{C}})\\
\leq & w_{\hat{r}}\sum_{i=1}^{\hat{r}}d_i(\B{X}\hat{\B{C}}-\B{XB}) + \sum_{i=1}^{\hat{r}} (w_{\hat{r}}-w_i)d_i(\B{X}\hat{\B{C}}-\B{XB})\\
\leq &\left\{2d^{-\gamma}_{\hat{r}}(\B{PY})-d_1^{-\gamma}(\B{PY})\right\}\sum_{i=1}^{\hat{r}}d_i(\B{X}\hat{\B{C}}-\B{XB})\\
\leq &\left\{2(d_{\hat{r}}(\B{XC})-d_1(\B{PE}))^{-\gamma}-(d_{1}(\B{XC})+d_1(\B{PE}))^{-\gamma}\right\}\sum_{i=1}^{\hat{r}}d_i(\B{X}\hat{\B{C}}-\B{XB})\\
\leq &\left\{2(d_{\hat{r}}(\B{XC})-d_1(\B{PE}))^{-\gamma}-(d_{1}(\B{XC})+d_1(\B{PE}))^{-\gamma}\right\}\sqrt{\hat{r}}\|\B{X}\hat{\B{C}}-\B{XB}\|_F.
\end{align*}
The last inequality is due to the Cauchy-Schwarz inequality. Using (\ref{th2:eq1}), $r(\B{X}\hat{\B{C}}-\B{XB})\leq r(\hat{\B{C}}-\B{B})\leq 2\hat{r}$ and the inequality $2xy\leq x^2/a+ay^2$ we have
\begin{align*}
\|\B{X}\hat{\B{C}}-\B{XC}\|_F^2
\leq& \|\B{X}\B{B}-\B{XC}\|_F^2+a\|\B{X}\hat{\B{C}}-\B{XB}\|_F^2\\
&+\frac{1}{a}\left\{d_1(\B{PE})\sqrt{2\hat{r}}+2\lambda(d_{\hat{r}}(\B{XC})-d_1(\B{PE}))^{-\gamma}\sqrt{\hat{r}}-
\lambda(d_{1}(\B{XC})+d_1(\B{PE}))^{-\gamma}\sqrt{\hat{r}}\right\}^2
\end{align*}
Since $\|\B{X}\hat{\B{C}}-\B{XB}\|_F^2\leq \|\B{X}\hat{\B{C}}-\B{XC}\|_F^2+ \|\B{X}\B{B}-\B{XC}\|_F^2$, consequently, for any $0<a<1$,
\begin{align*}
\|\B{X}\hat{\B{C}}-\B{XC}\|_F^2
\leq& \frac{1+a}{1-a}\|\B{X}\B{B}-\B{XC}\|_F^2\\
&+\frac{1}{a(1-a)}\left\{\sqrt{2}d_1(\B{PE})+2\lambda(d_{\hat{r}}(\B{XC})-d_1(\B{PE}))^{-\gamma}-\lambda(d_{1}(\B{XC})+d_1(\B{PE}))^{-\gamma}\right\}^2\hat{r}.
\end{align*}
As shown in Theorem \ref{th:rank}, on the event $\{d_1(\B{PE})< \delta\lambda^{\frac{1}{\gamma+1}}\}$, the estimated rank $\hat{r}$ equals to the true rank $r^*$, i.e., $\hat{r}=r^*$, and $\M{P}\{d_1(\B{PE})\geq \delta\lambda^{\frac{1}{\gamma+1}}\}\leq \exp(-\theta^2(r_x+q)/2)$. Also, $d_{r^*}(\B{XC})>2\lambda^{\frac{1}{\gamma+1}}$ and $c=d_{1}(\B{XC})/d_{r^*}(\B{XC})\geq 1$. Therefore, with probability at least $1-\exp(-\theta^2(r_x+q)/2)$,
\begin{align*}
\|\B{X}\hat{\B{C}}-\B{XC}\|_F^2
\leq& \frac{1+a}{1-a}\|\B{X}\B{B}-\B{XC}\|_F^2
+\frac{1}{a(1-a)}\left\{\sqrt{2}\delta\lambda^{\frac{1}{\gamma+1}}+2\lambda(2-\delta)^{-\gamma}\lambda^{\frac{-\gamma}{\gamma+1}}-\lambda(2c+\delta)^{-\gamma}\lambda^{\frac{-\gamma}{\gamma+1}}\right\}^2r^*\\
\leq& \frac{1+a}{1-a}\|\B{X}\B{B}-\B{XC}\|_F^2
+\frac{1}{a(1-a)}\left\{\sqrt{2}\delta+2(2-\delta)^{-\gamma}-(2c+\delta)^{-\gamma}\right\}^2\lambda^{\frac{2}{\gamma+1}}r^*.
\end{align*}
Since $\B{B}$ is an arbitrary matrix with $r(\B{B})\leq r^*$, the second part of the theorem is obtained by taking $\B{B}=\B{C}$ and $a=1/2$. This completes the proof.

\end{proof}

\subsection*{Proof of Corollary \ref{corollary:2}}

\begin{proof}
Consider any $\B{B}$ with $r(\B{B})\leq \hat{r}$, we have
\begin{align*}
&\sum w_i d_i(\B{XB})-\sum w_i d_i(\B{X}\hat{\B{C}})\\
=&w_{\hat{r}}\sum_{i=1}^{\hat{r}}d_i(\B{XB})-w_{\hat{r}}\sum_{i=1}^{\hat{r}}d_i(\B{X}\hat{\B{C}})
+\sum_{i=1}^{\hat{r}} (w_{\hat{r}}-w_i)d_i(\B{X}\hat{\B{C}})-\sum_{i=1}^{\hat{r}} (w_{\hat{r}}-w_i)d_i(\B{XB}).
\end{align*}
Note that $w_{\hat{r}}-w_1\geq\cdots\geq w_{\hat{r}}-w_{\hat{r}-1}\geq0$, so both $p_1(\cdot)=\sum_{i=1}^{\hat{r}}d_i(\cdot)$ and $p_2(\cdot)=\sum_{i=1}^{\hat{r}}(w_{\hat{r}}-w_i)d_i(\cdot)$ satisfy the triangular inequality. See the proof of Theorem \ref{thm:convexity}. It then follows that
\begin{align*}
\sum w_i d_i(\B{XB})-\sum w_i d_i(\B{X}\hat{\B{C}})
\leq&w_{\hat{r}}\sum_{i=1}^{\hat{r}}d_i(\B{X}\hat{\B{C}}-\B{XB}) + \sum_{i=1}^{\hat{r}} (w_{\hat{r}}-w_i)d_i(\B{X}\hat{\B{C}}-\B{XB})\\
=&\sum_{i=1}^{\hat{r}} (2w_{\hat{r}}-w_i)d_i(\B{X}\hat{\B{C}}-\B{XB})\\
\leq& \sqrt{\sum_{i=1}^{\hat{r}}(2w_{\hat{r}}-w_i)^2}\|\B{X}\hat{\B{C}}-\B{XB}\|_F.
\end{align*}
The last inequality is due to the Cauchy-Schwarz inequality.

Using (\ref{th2:eq1}) and the inequality $2xy\leq x^2/a+ay^2$ we have
\begin{align*}
&\|\B{X}\hat{\B{C}}-\B{XC}\|_F^2\notag\\
\leq& \|\B{X}\B{B}-\B{XC}\|_F^2
+\frac{1}{a}\left\{d_1(\B{PE})\sqrt{r(\B{X}\hat{\B{C}}-\B{XB})}+\lambda\sqrt{\sum_{i=1}^{\hat{r}}(2w_{\hat{r}}-w_i)^2}\right\}^2
+a\|\B{X}\hat{\B{C}}-\B{XB}\|_F^2\\
\leq& (1+a)\|\B{X}\B{B}-\B{XC}\|_F^2
+\frac{1}{a}\left\{d_1(\B{PE})\sqrt{r(\B{X}\hat{\B{C}}-\B{XB})}+\lambda\sqrt{\sum_{i=1}^{\hat{r}}(2w_{\hat{r}}-w_i)^2}\right\}^2
+a\|\B{X}\hat{\B{C}}-\B{XC}\|_F^2.
\end{align*}
Consequently, for any $0<a<1$,
\begin{align*}
\|\B{X}\hat{\B{C}}-\B{XC}\|_F^2
\leq& \frac{1+a}{1-a}\|\B{X}\B{B}-\B{XC}\|_F^2
+\frac{1}{a(1-a)}\left\{d_1(\B{PE})\sqrt{2\hat{r}}+\lambda\sqrt{\sum_{i=1}^{\hat{r}}(2w_{\hat{r}}-w_i)^2}\right\}^2.
\end{align*}
As shown in Theorem \ref{th:rank}, on the event $\{d_1(\B{PE})< \delta\lambda M\}$, the estimated rank $\hat{r}$ equals to the true rank $r^*$, i.e., $\hat{r}=r^*$, and $\M{P}\{d_1(\B{PE})\geq \delta\lambda M\}\leq \exp(-\theta^2(r_x+q)/2)$. Therefore, with probability at least $1-\exp(-\theta^2(r_x+q)/2)$,
\begin{align*}
\|\B{X}\hat{\B{C}}-\B{XC}\|_F^2
\leq& \frac{1+a}{1-a}\|\B{X}\B{B}-\B{XC}\|_F^2
+\frac{1}{a(1-a)}\left\{\lambda\delta M \sqrt{2r^*}+\lambda\sqrt{\sum_{i=1}^{r^*}(2w_{r^*}-w_i)^2}\right\}^2\\
\leq& \frac{1+a}{1-a}\|\B{X}\B{B}-\B{XC}\|_F^2
+\frac{1}{a(1-a)}\lambda^2r^*\{(2+\sqrt{2}\delta)M-w_1\}^2.
\end{align*}
Since $\B{B}$ is an arbitrary matrix with $r(\B{B})\leq \hat{r}$, the second part of the theorem is obtained by taking $\B{B}=\B{C}$ and $a=1/2$. This completes the proof.

\end{proof}

\bibliographystyle{biometrika}

\end{document}